\documentclass[12pt]{article}

\usepackage{ae,aecompl}

\usepackage{amsmath}
\usepackage{amssymb} 
\usepackage{amsthm} 
\usepackage{bm} 
\usepackage{commath} 
\usepackage{geometry}
\usepackage{graphicx} 
\usepackage{mathrsfs} 
\usepackage{mleftright} 
\mleftright 
\usepackage[authoryear]{natbib}
\usepackage{paralist} 
\usepackage{setspace}
\usepackage{verbatim} 
\usepackage{enumitem}
\usepackage{float}

\usepackage[backref=page, hyperindex=true]{hyperref}
\hypersetup{%
    breaklinks=true,
    colorlinks=true,
    linkcolor=magenta,
    urlcolor=blue,
    citecolor=blue,
    pdfstartview={FitH},
    unicode=true
 }

\let \backreforig \backref
\renewcommand*{\backref}[1]{[\backreforig{#1}]}

\theoremstyle{plain}
\newtheorem{assumption}{Assumption}
\numberwithin{assumption}{section}
\theoremstyle{plain}
\newtheorem{lem}{Lemma}
\numberwithin{lem}{section}
\theoremstyle{plain}
\newtheorem{thm}{Theorem}
\numberwithin{thm}{section}
\theoremstyle{definition}
\newtheorem{rem}{Remark}
\numberwithin{rem}{section}

\def\E{\mathrm{E}}

\def\R{\mathbb{R}}
\def\max{\mathrm{max}}
\def\tr{\mathrm{tr}}
\def\var{\mathrm{var}}

\def\bzero{\boldsymbol{0}}
\def\bone{\boldsymbol{1}}

\def\bm{\boldsymbol{m}}

\def\bx{\boldsymbol{x}}
\def\bX{\boldsymbol{X}}

\def\bY{\boldsymbol{Y}}

\def\ba{\boldsymbol{a}}

\def\bZ{\boldsymbol{Z}}

\def\bbA{\mathbf{A}}
\def\bbB{\mathbf{B}}
\def\bI{\mathbf{I}}

\def\cN{\mathcal{N}}

\def\bh{\boldsymbol{h}}
\def\bmu{\boldsymbol{\mu}}
\def\bet{\boldsymbol{\eta}}

\def\btheta{\boldsymbol{\theta}}
\def\bgamma{\boldsymbol{\gamma}}
\def\bphi{\boldsymbol{\phi}}
\def\bpsi{\boldsymbol{\psi}}

\def\bSigma{\boldsymbol{\Sigma}}
\def\bOmega{\boldsymbol{\Omega}}
\def\bTheta{\boldsymbol{\Theta}}

\newcommand*{\argmin}{\operatornamewithlimits{argmin}\limits}
\def\vec{\mathrm{vec}}

\newlist{inlinenum}{enumerate*}{1}
\setlist[inlinenum]{label=(\arabic*), ref=\arabic*}

\usepackage[affil-it]{authblk} 

\begin{document}

\title{Triple/Double-Debiased Lasso\thanks{We thank Jin Hahn, Whitney Newey, and Andres Santos for their helpful comments. S{\o}rensen is grateful for research support from the Aarhus Center for
Econometrics (ACE) funded by the Danish National Research Foundation grant
number DNRF186.
}}

\author[1]{Denis Chetverikov}
\author[2,4]{Jesper R.-V.~S{\o}rensen}
\author[3]{Aleh Tsyvinski}
\affil[1]{\small{University of California, Los Angeles}}
\affil[2]{\small{University of Copenhagen}}
\affil[3]{\small{Yale University}}
\affil[4]{\small{Aarhus Center for Econometrics (ACE)}}

\date{\today}

\maketitle

\begin{abstract}
In this paper, we propose a triple (or double-debiased) Lasso estimator for inference on a low-dimensional parameter in high-dimensional linear regression models. The estimator is based on a moment function that satisfies not only first- but also second-order Neyman orthogonality conditions, thereby eliminating both the leading bias and the second-order bias induced by regularization. We derive an asymptotic linear representation for the proposed estimator and show that its remainder terms are never larger and are often smaller in order than those in the corresponding asymptotic linear representation for the standard double Lasso estimator. Because of this improvement, the triple Lasso estimator often yields more accurate finite-sample inference and confidence intervals with better coverage. Monte Carlo simulations confirm these gains. In addition, we provide a general recursive formula for constructing higher-order Neyman orthogonal moment functions in Z-estimation problems, which underlies the proposed estimator as a special case.
\end{abstract}


\section{Introduction}

Inference on a low-dimensional target parameter in the presence of many controls has become a central problem in econometrics because making a conditional exogeneity assumption plausible in empirical applications often requires using a rich set of covariates \citep{BCFH17}. In such settings, in addition to the target parameter, we also have a high-dimensional nuisance parameter that has to be estimated via machine learning methods such as the Lasso estimator. Unfortunately, regularization underlying these machine learning methods generates non-regular first-stage bias that typically invalidates naive plug-in inference procedures for the target parameter. A major insight of the double/debiased machine learning literature is that this problem can often be overcome by replacing the original moment equation with the one based on a moment function whose first derivative with respect to the nuisance parameter vanishes at the true parameter values (first-order Neyman orthogonality condition), so that sufficiently accurate regularized estimators of the nuisance parameter matter only through second-order remainder terms and the target parameter can be estimated at the usual $1/\sqrt{n}$ rate \citep{CHS15}. In this paper we show that, in the high-dimensional linear regression model, one can go one step further and construct an estimator that removes not only the first-order effect of nuisance parameter estimation but also the second-order effect. Because of the steps in the construction of our estimator, we refer to it as a \emph{triple Lasso} estimator, or, equivalently, a \emph{double-debiased Lasso} estimator.

Our starting point is the standard high-dimensional linear regression model in which the outcome is regressed on a scalar treatment variable of interest and a potentially high-dimensional vector of controls. In this model, the conventional double Lasso estimator \citep{BCH14} is based on a moment function that satisfies the first-order Neyman orthogonality condition and is often remarkably successful even when the underlying Lasso estimators are far from being $\sqrt{n}$-consistent. However, the moment function behind the double Lasso estimator does not satisfy the {\em second-order} Neyman orthogonality condition, and this means that the double Lasso estimator may perform poorly if second-order remainder terms are large enough to matter in empirically relevant samples, which happens when the first-stage Lasso estimators converge rather slowly. Our first main contribution is thus to construct a moment function that satisfies not only the first-order but also the second-order Neyman orthogonality condition. Our moment function is obtained from the double Lasso moment function by subtracting a carefully chosen quadratic correction term, weighted by the inverse Gram matrix of the controls, so that the second-order bias term is cancelled by construction. In turn, our second main contribution is to propose the triple Lasso estimator that solves a moment condition based on this moment function, with all underlying nuisance parameters being estimated by the Lasso-type methods.

We derive an asymptotic linear representation for our triple Lasso estimator and compare its remainder terms with those of the usual double Lasso estimator. We show that whenever the double Lasso estimator is already asymptotically normal, the remainder terms for the triple Lasso estimator are never larger and are often strictly smaller in order than the corresponding terms for the double Lasso estimator. As a result, the distribution of the triple Lasso estimator is often closer to the properly centered normal distribution than that of the double Lasso estimator. Moreover, in the special case of exact sparsity and consistent Lasso screening, the triple Lasso estimator can remain asymptotically normal along sequences of data-generating processes for which the double Lasso estimator fails to be asymptotically normal. 

Our Monte Carlo simulations confirm that the implications of the asymptotic theory remain relevant in empirically meaningful finite-sample settings. Across designs that vary the sample size, the correlation structure of the controls, and the sparsity of the nuisance parameters, although the variance of the triple Lasso estimator is typically somewhat larger than that of the double Lasso estimator, the former estimator often has substantially smaller squared bias and the mean squared error, especially in the harder designs. More importantly, however, the triple Lasso estimator often leads to more precise inference. For example, in one of the designs, the empirical coverage of nominal $95\%$ confidence intervals rises from about $67\%$ under the double Lasso estimator to about $92\%$ under the triple Lasso estimator. This happens despite the fact that the triple Lasso confidence intervals are only slightly longer on average.

As our third main contribution, we extend the construction of the second-order Neyman orthogonal moment function to a general Z-estimation problem. In fact, we provide a general recursive formula that can be used to construct moment functions in the Z-estimation problem that satisfy the Neyman orthogonality condition to {\em any} desired order. We further illustrate this construction in an M-estimation problem with a single-index structure, delivering the results for the high-dimensional linear regression model as a special case.


Our paper is related to several strands of the literature. First, it builds directly on the foundational work on double/debiased Lasso and post-selection inference in high-dimensional linear regression models, including \cite{BCH14}, \cite{JM14}, \cite{GBRD14}, and \cite{ZZ14}, all of which established that valid inference on low-dimensional components is possible after regularization when the moment function is appropriately debiased. Closely related is the double machine learning literature of \cite{CHS15} and \cite{CCDDHNR18}, from which we adopt the emphasis on orthogonal moment functions, while extending that perspective from first-order to higher-order orthogonality. An alternative approach to inference in high-dimensional linear regression models is developed by \cite{AKK20}, who propose explicitly accounting for the bias induced by regularization rather than eliminating it through orthogonalization; an advantage of this approach is that it achieves validity under weaker structural conditions, although it requires the researcher to specify a bound on the magnitude of the nuisance coefficients. Second, our paper is connected to the broader literature on approximately sparse high-dimensional models, including \cite{BC11}, \cite{belloni_sparse_2012}, and \cite{G16}, from which we borrow technical tools, in particular the convergence rates for Lasso estimators in linear regression models and for node-wise Lasso estimators of inverse Gram matrices. Third, our paper is related to the emerging literature on higher-order orthogonality, including \cite{MSZ18} and \cite{BJW24}. In particular, \cite{MSZ18} construct higher-order Neyman orthogonal moment functions for the partially linear model under {\em non}-Gaussian first-stage errors, while \cite{BJW24} develop a general algorithm for constructing higher-order Neyman orthogonal moment functions in likelihood-based models. These results are not directly applicable in our setting, as we focus on regression-based estimation rather than likelihood models and do not restrict the distribution of noise. Also, our results do not contradict the negative result in \cite{MSZ18} on the existence of higher-order Neyman orthogonal moment functions in the partially linear model with Gaussian first-stage errors, because our moment function depends not only on a vector of observable random variables but also on an independent copy of this vector, and, moreover, our orthogonality requirement is formulated in terms of ordinary derivatives rather than the stronger functional-derivative notion studied there, reflecting the parametric nature of our regression model. Finally, our paper is also related to earlier work on higher-order influence functions and U-statistic-based estimators in semiparametric and nonparametric models, including \cite{RLTV08}, \cite{RLMTV17}, \cite{NR18}, and \cite{LMNR23}, which develop higher-order expansions of functionals and corresponding estimators that balance bias and variance in settings where first-order influence function methods are insufficient. This literature shows that higher-order corrections can be used to relax smoothness requirements and/or improve rates in semiparametric and nonparametric models.

The rest of the paper is organized as follows. Section \ref{sec: triple lasso} introduces the triple Lasso moment function and the corresponding estimator in the high-dimensional linear regression model, proves second-order Neyman orthogonality of the proposed moment function, and derives the asymptotic linear representation and asymptotic normality result for the triple Lasso estimator. Section \ref{sec: extensions} develops the general recursive formula for constructing higher-order Neyman orthogonal moment functions for the Z-estimation problem and applies it to M-estimators with a single-index structure, thereby showing that the logic behind the triple Lasso estimator is a part of a broader higher-order orthogonalization principle. Section \ref{sec: simulations} reports Monte Carlo results showing that the triple Lasso estimator can substantially reduce bias and deliver more reliable inference in comparison with the double Lasso estimator under certain data-generating processes.

\subsection*{Notation}
For any $p\in\mathbb N$, we denote $[p]:=\{1,\dots,p\}$ and $\bzero_p := (0,\dots,0)^\top \in\R^p$. Also, we let $\bzero_{p\times p}$ and $\bI_p$ be the zero and identity matrices in $\R^{p\times p}$, respectively. In addition, for any finite set of positive integers $I$, we use $\E_I[f(\bX_i)]$ to denote the average value of $f(\bX_i)$ as $i$ varies over $I$, namely $\E_I[f(\bX_i)] = |I|^{-1}\sum_{i\in I}f(\bX_i)$. For numbers $a_n$ and $b_n$, $n\in\mathbb N$, we write $a_n\lesssim b_n$ if there exists a constant $C>0$ such that $|a_n|\leq C|b_n|$ for all $n\in \mathbb N$. For random variables $V_n$ and $R_n$, $n\in\mathbb N$, we write $V_n\lesssim_P R_n$ if $V_n = Y_n R_n$ for some $Y_n = O_P(1)$. Finally, for any matrix $\bbA = (A_{j,k})_{j,k=1}^p$, we denote $\|\bbA\|_{\max} := \max_{1\leq j,k\leq p}|A_{j,k}|$ and $\|\bbA\|_{\infty} := \max_{1\leq i\leq p}\sum_{j=1}^p |A_{j,k}|$.

\section{Triple Lasso Estimator}\label{sec: triple lasso}
In this section, we propose a moment function for the high-dimensional linear regression model that satisfies the second-order Neyman orthogonality condition and introduce the corresponding triple Lasso estimator. In addition, we derive the $\sqrt n$-consistency and asymptotic normality result for our triple Lasso estimator and compare the remainder terms of the corresponding asymptotic linear representation to those underlying the asymptotic linear representation for the double Lasso estimator.
\subsection{Model, Second-Order Neyman Orthogonality, and Estimation}
Consider the linear regression model
\begin{equation}\label{eq: model}
Y = D\beta_0 + \bX^\top \btheta_0 + \varepsilon,\quad \E\left[\varepsilon\middle| D,\bX\right]=0,
\end{equation}
where $Y\in\R$ is the outcome, $D\in\R$ is a regressor of interest, $\bX = (X_1,\dots,X_p)^\top\in\R^p$ is a vector of controls, $\beta_0\in\R$ is the parameter of interest, and $\btheta_0 = (\theta_{0,1},\dots,\theta_{0,p})^\top\in\R^p$ is a vector of nuisance parameters. Let $(\bX_i,D_i,Y_i)$, $i=1,\dots,n$, be a random sample from the distribution of $(\bX,D,Y)$, where we denote $\bX_i = (X_{i,1},\dots,X_{i, p})^\top$ for all $i\in[n]$. In this paper, we primarily focus on the high-dimensional setting in which the number of controls $p$ may be comparable to or larger than the sample size $n$.

A conventional way to estimate $\beta_0$ in model \eqref{eq: model} when $p$ is large is via the double Lasso estimator.\footnote{Various versions of the double/debiased Lasso estimator appeared in \cite{BCH14}, \cite{JM14}, \cite{GBRD14}, and \cite{ZZ14}.} To describe it, let the projection of $D$ on $\bX$ be
\begin{equation}\label{eq:first-stage}
D = \bX^\top \bgamma_0 + \nu, 
\quad \E[\nu\bX] = \bzero_p ,
\end{equation}
where $\bgamma_0 = (\gamma_{0,1},\dots,\gamma_{0,p})^\top\in\R^p$ is a vector of parameters, and let the reduced-form regression of $Y$ on $\bX$ be
\begin{equation}\label{eq:reduced-form}
Y = \bX^\top \bphi_0 + e,
\quad \E[e\bX] = \bzero_p ,
\end{equation}
where $\bphi_0 = (\phi_{0,1},\dots,\phi_{0,p})^\top := \btheta_0 + \beta_0 \bgamma_0$ and $e := \varepsilon + \beta_0\nu$.
One version of the double Lasso estimator is constructed as follows. 
First, estimate $\bphi_0$ by running the Lasso regression of $Y$ on $\bX$, yielding $\widehat\bphi$. 
Second, estimate $\bgamma_0$ by running the Lasso regression of $D$ on $\bX$, yielding $\widehat\bgamma$. 
Third, estimate $\beta_0$ by running the OLS regression of $Y-\bX^\top\widehat\bphi$ on $D-\bX^\top\widehat\bgamma$, yielding the double Lasso estimator $\widehat\beta^{DL}$. Under certain conditions, the double Lasso estimator is $\sqrt n$-consistent and asymptotically normal, despite the fact that the underlying Lasso estimators $\widehat\bphi$ and $\widehat\bgamma$ are not $\sqrt n$-consistent.

To understand why this happens, introduce the moment function $\psi^{DL}\colon \R\times \R^p\times \R^p$ by
\[
\psi^{DL}(\beta, \bgamma, \bphi) := \E[(Y - \bX^\top\bphi - \beta(D - \bX^\top\bgamma))(D - \bX^\top\bgamma)]
\]
and observe that the true value $\beta_0$ satisfies the moment equation
\begin{equation}\label{eq: moment condition}
\psi^{DL}(\beta,\bgamma_0,\bphi_0) = 0.
\end{equation}
The double Lasso estimator $\widehat\beta^{DL}$ solves an empirical analogue of this equation obtained by replacing the population expectation with the empirical expectation and the true values $\bgamma_0$ and $\bphi_0$ with their estimators $\widehat\bgamma$ and $\widehat\bphi$.
The key fact underlying the $\sqrt n$-consistency and asymptotic normality of the double Lasso estimator is that there is no first-order effect of replacing $\bgamma_0$ and $\bphi_0$ by $\widehat\bgamma$ and $\widehat\bphi$ on the moment function $\psi^{DL}$:
\begin{equation}\label{eq: neyman orthogonality}
\frac{\partial\psi^{DL}(\beta_0,\bgamma_0,\bphi_0)}{\partial\bgamma} = \bzero_p
\quad\text{and}\quad
\frac{\partial\psi^{DL}(\beta_0,\bgamma_0,\bphi_0)}{\partial\bphi} = \bzero_p.
\end{equation}
Indeed, provided that the estimators $\widehat\bgamma$ and $\widehat\bphi$ are sufficiently precise so that the second-order effects are asymptotically negligible, condition \eqref{eq: neyman orthogonality} ensures that the estimator $\widehat\beta^{DL}$ is asymptotically equivalent to the infeasible estimator $\widetilde\beta^{DL}$ obtained by solving the empirical analogue of \eqref{eq: moment condition} in which only the population expectation is replaced by the empirical expectation. The latter estimator is in turn $\sqrt n$-consistent and asymptotically normal by the classic $M$-estimation theory.

Because of \eqref{eq: neyman orthogonality}, the moment function $\psi^{DL}$ is said to satisfy the (first-order) {\em Neyman orthogonality} condition. To express this condition more compactly, let $\bet := (\bgamma^\top,\bphi^\top)^\top$ and $\bet_0 := (\bgamma_0^\top,\bphi_0^\top)^\top$. Then \eqref{eq: neyman orthogonality} can be written equivalently as
\[
\frac{\partial \psi^{DL}(\beta_0,\bet_0) }{\partial \bet}= \bzero_{2p}.
\]
Similarly, one can define the second-order Neyman orthogonality. Following \cite{MSZ18} and \cite{BJW24}, we say that a moment function $\psi\colon\R\times\R^q \to\R$ satisfies the {\em second-order Neyman orthogonality} condition if
\[
\frac{\partial \psi(\beta_0,\bet_0)}{\partial\bet} =\bzero_q
\quad\text{and}\quad
\frac{\partial^2 \psi(\beta_0,\bet_0)}{\partial\bet\partial\bet^\top}=\bzero_{q\times q},
\]
where $\bet_0\in\R^q$ is a vector of nuisance parameters. Intuitively, using estimators based on moment functions satisfying the second-order Neyman orthogonality condition is beneficial because such moment functions eliminate not only the first-order effects of replacing the true values of the nuisance parameters by the corresponding estimators but also the second-order effects. Unfortunately, however, the moment function $\psi^{DL}$ underlying the double Lasso estimator does not have this property. Indeed, it is straightforward to verify that
\[
\frac{\partial^2\psi^{DL}(\beta_0,\bgamma_0,\bphi_0)}{\partial\bgamma\partial\bphi}
=
\E[\bX\bX^\top]
\neq
\bzero_{p\times p}.
\]
It therefore follows that convergence of the double Lasso estimator to the properly centered normal distribution may fail or be slow if the Lasso estimators $\widehat\bgamma$ and $\widehat\bphi$ are not sufficiently precise and the second-order effects are non-negligible.

Motivated by this observation, we propose a novel moment function for estimating $\beta_0$ in model \eqref{eq: model} that satisfies not only the first-order Neyman orthogonality condition but also the second-order Neyman orthogonality condition. We then define a {\em triple Lasso} estimator that solves an empirical analogue of the moment condition implied by this function and derive its asymptotic properties.

In order to introduce our moment function, let $\bSigma_0 := \E[\bX\bX^\top]$ denote the population Gram matrix of the controls $\bX$, and let $\bTheta_0 := \bSigma_0^{-1}$ denote its inverse. Define the function $\psi^{ADJ} \colon \R\times \R^p\times \R^p \times \R^{p\times p}$ by
\[
\psi^{ADJ}(\beta, \bgamma, \bphi, \bTheta) := \E[(Y - \bX^\top\bphi - \beta(D - \bX^\top\bgamma))\bX^\top \bTheta \tilde\bX(\tilde D-\tilde \bX^\top\bgamma)],
\]
where $(\tilde\bX,\tilde D)$ is a copy of $(\bX,D)$ that is independent of $(\bX,D,Y)$. Our moment function $\psi^{TL}\colon \R\times\R^{2p+p^2}\to\R$ for estimating $\beta_0$ in model \eqref{eq: model} is then defined as
\begin{equation}\label{eq: double orthogonal score}
\psi^{TL}(\beta, \bet) := \psi^{DL}(\beta, \bgamma, \bphi) - \psi^{ADJ}(\beta, \bgamma, \bphi, \bTheta),\quad\bet: = (\bgamma^\top,\bphi^\top,\vec(\bTheta)^\top)^\top.
\end{equation}
In the following lemma, we show that the true value $\beta_0$  solves the moment equation
\begin{equation}\label{eq: moment condition triple lasso}
\psi^{TL}(\beta, \bet_0) = 0,
\end{equation}
where $\bet_0 := (\bgamma^\top_0,\bphi^\top_0,\vec(\bTheta_0)^\top)^\top$ and that the moment function $\psi^{TL}$ satisfies the second-order Neyman orthogonality condition. The proof of this lemma can be found in the Appendix.

\begin{lem}\label{lem: novel score}
As long as the moments $\E[Y^2]$, $\E[D^2]$, and $\E[\|\bX\|_2^2]$ are finite and the matrix $\E[\bX\bX^\top]$ is invertible, we have $\psi^{TL}(\beta_0, \bet_0) = 0$ and all first- and second-order derivatives of the function $\bet\mapsto \psi^{TL}(\beta_0, \bet)$ at $\bet = \bet_0$ are zero.
\end{lem}

\begin{rem}
It is interesting to compare Lemma \ref{lem: novel score} with a negative result on the existence of higher-order Neyman orthogonal moment functions in \cite{MSZ18}, abbreviated as MSZ below. They consider the closely related partially linear model
\begin{align*}
Y &= D\beta_0 + f_0(\bX) + \varepsilon, \qquad \E\left[\varepsilon\middle| D,\bX\right]=0,\\
D &= g_0(\bX) + \nu, \qquad \E\left[\nu\middle| \bX\right]=0,
\end{align*}
and show that if the conditional distribution of $\nu$ given $\bX$ is Gaussian, then any twice differentiable score function $m$ that is second-order Neyman orthogonal with respect to $f_0$ and $g_0$ must satisfy
\[
\frac{\partial \E[m(\bZ,\beta_0,\eta_0)]}{\partial\beta} = 0,
\]
where $\bZ:=(\bX,D,Y)$ and $\eta_0$ includes the nuisance functions $f_0$ and $g_0$. The latter in turn essentially rules out $\sqrt n$-consistent estimation of $\beta_0$ based on the moment condition $\E[m(\bZ,\beta_0,\eta_0)]=0$.
In contrast, for our high-dimensional linear regression model, we do have
\[
\frac{\partial\psi^{TL}(\beta_0,\bet_0)}{\partial\beta}
=
- \E[(D-\bX^\top\bgamma_0)^2]
\neq 0.
\]
To resolve this apparent contradiction, we make two observations. First, the score function underlying our moment function $\psi^{TL}$ depends not only on $(\bX,D,Y)$ but also on an independent copy $(\tilde\bX,\tilde D)$ of $(\bX,D)$. Second, our second-order Neyman orthogonality condition is weaker than that imposed in MSZ: we only require certain {\em ordinary} derivatives to vanish, whereas MSZ require certain {\em functional} derivatives to vanish, reflecting the parametric nature of our model.
\qed
\end{rem}

\begin{rem}
It is also interesting to compare the moment function $\psi^{TL}$ with what can be obtained from the algorithm of \cite{BJW24}, abbreviated as BJW below. Since their algorithm applies only to likelihood models and cannot, in general, be used directly in regression settings, we embed our model \eqref{eq: model}-\eqref{eq:first-stage} into a likelihood framework by imposing additional distributional assumptions on the vector of controls $\bX$ and the disturbances $(\varepsilon,\nu)$. In particular, for the purposes of this comparison, we assume that the distribution of $\bX$ is known and that the conditional distribution of $(\varepsilon,\nu)$ given $\bX$ is standard normal. Under these assumptions, our regression model reduces to a likelihood model, and applying the BJW algorithm yields the moment function $\psi^{BKW}\colon\R\times\R^{2p}\to\R$ defined by
$$
\psi^{BKW}(\beta,\bet):= \E[(1 - \E[\bm^\top](\E[\bm\bm^{\top}])^{-1}\bm)(Y - D\beta - \bX^\top\btheta)(D-\bX^\top\bgamma)],\ \bet:=(\bgamma^\top,\btheta^\top)^\top,
$$
where $\bm:=\mathrm{vech}(\bX\bX^\top)$, where the operator $\mathrm{vech}$ stacks the elements of the lower triangular part, including the diagonal, of the matrix into a vector. By construction, this moment function satisfies the second-order Neyman orthogonality condition under these distributional assumptions. However, if we relax these assumptions and only maintain $\E[\varepsilon\mid\bX] = 0$ and $\E[\nu\bX]=\bzero_p$ as in model \eqref{eq: model}-\eqref{eq:first-stage}, then the derivative $\partial\psi^{BKW}(\beta_0,\bet_0)/\partial\btheta$ need not vanish, so that $\psi^{BKW}$ may fail to satisfy even the first-order Neyman orthogonality condition. 

On the other hand, and perhaps more interestingly, if we strengthen the assumptions to $\E[\varepsilon\mid\bX] = 0$ and $\E[\nu\mid\bX] = 0$, then the moment function $\psi^{BKW}$ satisfies the second-order Neyman orthogonality condition not only with respect to $\bet=(\bgamma^\top,\btheta^\top)^\top$, but also with respect to the extended vector 
\[
\tilde\bet:=(\bgamma^\top,\btheta^\top,\E[\bm^\top],\mathrm{vech}(\E[\bm\bm^\top])^\top)^{\top},
\]
which is a desirable feature. However, using this moment function in practice requires estimating and inverting the $(p(p+1)/2) \times (p(p+1)/2)$ matrix $\E[\bm\bm^\top]$, which may be hard and may require non-standard conditions. In contrast, using our moment function $\psi^{TL}$ only requires estimating and inverting the much smaller $p\times p$ matrix $\E[\bX\bX^\top]$.
\qed
\end{rem}

We next define the triple Lasso estimator as a solution to an empirical analogue of equation \eqref{eq: moment condition triple lasso}, where $\bgamma_0$ and $\bphi_0$ are replaced by the corresponding Lasso estimators and $\bTheta_0$ is replaced by the node-wise Lasso estimator. For the reasons that will become clear shortly, we in fact use row-sparsified version of the node-wise Lasso estimator of $\bTheta_0$. In addition, since it is useful for the asymptotic normality result, we also combine the estimator with cross-fitting, which is standard in the literature on double machine learning; see \cite{CCDDHNR18}. 

To formally describe the resulting estimator, split at random the full sample into $K$ subsamples of approximately the same size for some small $K$ and let $I(1),\dots,I(K)$ be the sets of indices from $1$ to $n$ corresponding to the observations in these subsamples. Also, for each $k\in[K]$, let $I(-k) := [n]\setminus I(k)$. Then for each $k\in[K]$, let $\widehat\bgamma_k = (\widehat\gamma_{k,1},\dots,\widehat\gamma_{k,p})^\top$, $\widehat\phi_k = (\widehat\phi_{k,1},\dots,\widehat\phi_{k,p})^\top$, and $\widetilde\bTheta_k = (\widetilde\Theta_{k, j, l})_{j,l\in[p]}$ be the Lasso estimator of $\bgamma_0$, the Lasso estimator of $\bphi_0$, and the node-wise Lasso estimator of $\bTheta_0$, respectively, all defined on the subsample of observations with indices in $I(-k)$. Here, we use the row-wise version of the node-wise Lasso estimator. Specifically, to define the first row $\widetilde\bTheta_{k,1}$ of the matrix $\widetilde\bTheta_k$, we run the Lasso regression of $X_1$ on $\bX_{-1} := (X_2,\dots,X_p)^\top$ and obtain the vector of slope coefficients, say $\widehat\bpsi_1$. Then we set $\widehat\sigma_1^2 := \E_{I(-k)}[X_{i,1}(X_{i,1} - \bX_{i,-1}^\top\widehat\bpsi_1)]$ and $\widetilde\bTheta_{k,1} := (1, - \widehat\bpsi_1^\top)/\widehat\sigma_1^2$, where $\bX_{i,-1}:=(X_{i,2},\dots,X_{i,p})^\top$. All other rows of $\widetilde\bTheta_k$ are defined analogously. Next, let $\widehat T_{k,0} := \{j=1,\dots,p\colon \widehat\gamma_{k,j} \neq 0\}$ be the set of indices from $1$ to $p$ corresponding to the non-zero components of the vector $\widehat\bgamma_k$ and let $\widehat\bTheta_{k}$ be the $\R^{p\times p}$ matrix obtained from $\widetilde\bTheta_k$ by setting to zero all its rows whose indices are not in $\widehat T_k:=\widehat T_{k,0} \cup \widehat T_{k,1}$, where $\widehat T_{k,1}$ is an additional set of indices corresponding to the extra variables that the researcher might want to use as we discuss below. Throughout most of the paper, we assume that $\widehat T_{k,1} = \emptyset$. To define $\widehat\beta_k$, we now solve an empirical version of the equation \eqref{eq: moment condition triple lasso} on $I_k$ for $\beta$, namely we set
$$
\widehat\beta_k := \frac{\E_{I(k)}[(Y_i - \bX_i^\top \widehat\bphi_k)(D_i - \bX_i^\top\widehat\bgamma_k)] - \E_{I(k)}[(Y_i - \bX_i^\top\widehat\bphi_k)\bX_i^\top]\widehat\bTheta_{k}\E_{I(k)}[\bX_i(D_i - \bX_i^\top\widehat\bgamma_k)]}{\E_{I(k)}[(D_i - \bX_i^\top\widehat\bgamma_k)^2] - \E_{I(k)}[(D_i - \bX_i^\top\widehat\bgamma_k)\bX_i^\top]\widehat\bTheta_{k} \E_{I(k)}[\bX_i(D_i - \bX_i^\top\widehat\bgamma_k)]}.
$$
Finally, we aggregate the subsample estimators $\widehat\beta_k$ to obtain the triple (or double-debiased) Lasso estimator:
$$
\widehat\beta^{TL} := \frac{1}{K}\sum_{k=1}^K \widehat\beta_k.
$$
In the next subsection, we will derive the asymptotic theory for the triple Lasso estimator $\widehat\beta^{TL}$ and compare it with that for the double Lasso estimator $\widehat\beta^{DL}$.

\begin{rem}
We now explain why we use sparsified versions $\widehat\bTheta_k$ of the node-wise Lasso estimators $\widetilde\bTheta_k$ of $\bTheta_0$. Consider an idealized case where $\E[\bX\bX^\top]$ is equal to the identity matrix $I_p$ and this fact is known to the researcher. In this case, we could replace $\widehat\bTheta_k$ in the expression for $\widehat\beta_k$ above by $\bI_p$. Doing so would lead to the following term in the numerator of $\widehat\beta_k$:
$$
\E_{I_k}[(Y_i - \bX_i^\top\widehat\bphi_k)\bX_i^\top]\E_{I_k}[\bX_i(D_i - \bX_i^\top\widehat\bgamma_k)].
$$
Substituting here the expressions for $D_i$ and $Y_i$ from \eqref{eq:first-stage} and \eqref{eq:reduced-form}, respectively, yields terms of the form $\E_{I_k}[\varepsilon_i\bX_i^\top]\E_{I_k}[\bX_i\nu_i]$. In turn, such terms are difficult to control when $p$ is comparable to or larger than $n$. In particular, it is straightforward to verify that they are typically of order $\sqrt p/n$ in probability, which is too large for our purposes. By introducing sparsification, we substantially reduce the variance of these terms at the expense of introducing some bias.
\qed
\end{rem}

\begin{rem}
Although we focus on Lasso-type methods throughout this section, which lead to the triple Lasso estimator, the moment function $\psi^{TL}$ can in principle be combined with other machine learning methods suitable for estimating the nuisance parameters $\bgamma_0$, $\bphi_0$, and $\bTheta_0$. For instance, $\bgamma_0$ and $\bphi_0$ can be estimated using the Dantzig selector, while $\bTheta_0$ can be estimated via a Dantzig-type procedure as discussed, for example, in \cite{JM14}. In fact, the asymptotic theory developed in the next subsection is agnostic to the specific choice of estimators, provided they achieve the required level of precision as specified in our assumptions.
 \qed
\end{rem}

\subsection{Asymptotic Theory}
In this subsection, we establish the $\sqrt n$-consistency and asymptotic normality result for the triple Lasso estimator $\widehat\beta^{TL}$. To this end, we impose several regularity conditions. These conditions control the moments of the disturbances $(\varepsilon,\nu)$, the sample splitting scheme used in the construction of the estimator, the behavior of the Lasso estimators appearing in the algorithm, and the sparsity of the model. 

\begin{assumption}\label{as: moments}
\begin{inlinenum}
\item The random variables $\varepsilon$ and $\nu$ have bounded fourth moments: $\E[|\varepsilon|^4] \lesssim 1$ and $\E[|\nu|^4]\lesssim 1$; 
\item the second moment of $\nu$ is bounded below from zero: $(\E[\nu^2])^{-1}\lesssim 1$.
\end{inlinenum}
\end{assumption}

\begin{assumption}\label{as: bounded}
There exists a constant $C\in(0,\infty)$ such that $\E[\varepsilon^2\mid \bX]\leq C$, $\E[\nu^2\mid \bX]\leq C$, and $\|X\|_{\infty}\leq C$ almost surely.
\end{assumption}

\begin{assumption}\label{as: subsample}
For all $k\in[K]$, the subsample $I(k)$ contains a non-vanishing fraction of the full sample: $|I(k)|^{-1}\lesssim n^{-1}$.
\end{assumption}

Assumptions \ref{as: moments}--\ref{as: subsample} impose basic moment and regularity conditions on the disturbances, regressors, and the sample splitting scheme. Assumption \ref{as: moments} requires the disturbances $\varepsilon$ and $\nu$ to have bounded fourth moments and the variance of $\nu$ to be bounded away from zero. The latter serves as an identification condition ensuring that the leading term in the asymptotic expansion of the triple Lasso estimator has finite variance. Assumption \ref{as: bounded} imposes mild boundedness conditions on the conditional second moments of $\varepsilon$ and $\nu$ and requires the regressors to be uniformly bounded. These conditions simplify the derivations but we note that our results can be extended to allow for unbounded controls. We work with bounded controls to avoid technicalities and to keep the arguments transparent. Finally, Assumption \ref{as: subsample} formalizes the requirement that the subsamples used in cross-fitting contain a non-vanishing fraction of the observations in the full sample. Since the number of folds $K$ is fixed, this condition guarantees that the number of observations in each subsample is of order $n$.

\begin{assumption}\label{as: sparsity}
There are non-random sequences $s_{\bgamma} := s_{\bgamma,n}$ and $s_{\btheta} := s_{\btheta,n}$ of integers in $[1,\infty)$ and non-random sequences $\bar\bgamma_0:=\bar\bgamma_{0,n}$ and $\bar\btheta_0:=\bar\btheta_{0,n}$ of vectors in $\R^p$ such that:
\begin{enumerate}[label=(\arabic*), ref=\arabic*]
 \item the vectors $\bar\bgamma_0$ and $\bar\btheta_0$ are sparse: $\|\bar\bgamma_0\|_0\leq s_{\bgamma}$ and $\|\bar\btheta_0\|_0\leq s_{\btheta}$; \label{as: sparse vectors}
 
 \item the vectors $\bar\bgamma_0$ and $\bar\btheta_0$ provide good approximation to the parameter vectors $\bgamma_0$ and $\btheta_0$: $\| \bgamma_0 - \bar\bgamma_0 \|_2^2 \lesssim s_{\bgamma}/n$, $\| \bgamma_0 - \bar\bgamma_0 \|_1^2 \lesssim s_{\bgamma}^2/n$, $\E[\bX^\top(\bgamma_0 - \bar\bgamma_0)|^2] \lesssim s_{\bgamma}/n$, $\| \btheta_0 - \bar\btheta_0 \|_2^2 \lesssim s_{\btheta}/n$, $\| \btheta_0 - \bar\btheta_0 \|_1^2 \lesssim s_{\btheta}^2/n$, and $\E[\bX^\top(\btheta_0 - \bar\btheta_0)|^2] \lesssim s_{\btheta}/n$;\label{as: sparse approximations}

 \item for all $k\in[K]$, the estimators $\widehat\bgamma_k$ and $\widehat\bphi_k$ are sufficiently accurate for $\bar\bgamma_0$ and $\bar\bphi_0:=\bar\btheta_0 + \beta_0\bar\bgamma_0$: $\| \widehat\bgamma_k - \bar\bgamma_0 \|_2^2\lesssim_P s_{\bgamma}\log p / n$, $\| \widehat\bgamma_k - \bar\bgamma_0 \|_1^2\lesssim_P s_{\bgamma}^2\log p / n$, $\| \widehat\bphi_k - \bar\bphi_0 \|_2^2\lesssim_P s_{\bphi}\log p / n$, and $\| \widehat\bphi_k - \bar\bphi_0 \|_1^2\lesssim_P s_{\bphi}^2\log p / n$, where we denoted $s_{\bphi} := s_{\btheta} + s_{\bgamma}$;\label{as: lasso estimation}

 \item for all $k\in[K]$, the estimators $\widehat\bgamma_k$ and $\widehat\bphi_k$ are sufficiently sparse: $\|\widehat\bgamma_k\|_0 \lesssim_P s_{\bgamma}$ and $\|\widehat\bphi_k\|_0 \lesssim_P s_{\bphi}$; \label{as: sparse estimators}
 
 \item the matrix $\E[\bX\bX^\top]$ satisfies a sparse eigenvalue condition: $\lambda_{\max,s_{\bphi}\ell_n}(\E[\bX\bX^\top]) \lesssim 1$ for some $\ell_n\to\infty$ as $n\to\infty$; \label{as: sparse eigenvalue}
 
 \item for all $k\in[K]$, the sets $\widehat T_k$ are sufficiently sparse: $|\widehat T_k| \lesssim_P s_{\bgamma}$. \label{as: sparse choice}
 \end{enumerate}
\end{assumption}

Assumption \ref{as: sparsity} imposes approximate sparsity and regularity conditions that are standard in the analysis of Lasso estimators. Parts \ref{as: sparse vectors} and \ref{as: sparse approximations} require that the parameters $\bgamma_0$ and $\btheta_0$ admit sparse approximations $\bar\bgamma_0$ and $\bar\btheta_0$ with sparsity indices $s_{\bgamma}$ and $s_{\btheta}$ and formalize the quality of these approximations by requiring the approximation errors to be small in both $\ell_2$ and $\ell_1$ norms as well as in the prediction norm induced by the controls $\bX$. Comparable conditions appear, for example, in \cite{BCH14} and \cite{CHS15}. Parts \ref{as: lasso estimation} and \ref{as: sparse estimators} impose the usual rate and sparsity properties of the Lasso estimators $\widehat\bgamma_k$ and $\widehat\bphi_k$ obtained in the auxiliary regressions. These conditions are satisfied under typically used assumptions on the design matrix, as explained for example in \cite{BC11}, and ensure that the preliminary estimators are sufficiently accurate for the sparse approximations $\bar\bgamma_0$ and $\bar\bphi_0$. Part \ref{as: sparse eigenvalue} imposes a classical sparse eigenvalue condition on the population Gram matrix $\E[\bX\bX^\top]$, meaning that sparse subsets of the regressors are not excessively collinear. Finally, Part \ref{as: sparse choice} requires the sets $\widehat T_k$ used in the construction of the triple Lasso estimator to be sufficiently sparse. In light of Part \ref{as: sparse estimators}, this condition holds trivially if $\widehat T_{k,1} = \emptyset$. 

\begin{assumption}\label{as: nodewise lasso}
For all $k\in[K]$, the node-wise Lasso estimators $\widetilde\bTheta_k$ are such that: 
\begin{inlinenum}
\item $\allowbreak \|\widetilde\bTheta_k\E_{I(-k)}[\bX_i\bX_i^\top] - \bI_p\|_{\max}^2 \lesssim_P \log p / n$; \label{as: max approximation nodewise}
\item $\|\widetilde\bTheta_k\|_{\infty} \lesssim_P 1$.\label{as: bounded row nodewise}
\end{inlinenum}
\end{assumption}

\begin{assumption}\label{as: growth conditions}
We have 
\begin{inlinenum}
\item $\sqrt{s_{\bgamma}}s_{\bphi}\log p / n = o(1)$;
\item $s_{\bgamma}^3s_{\bphi}(\log p)^3 / n^2 = o(1)$.
\end{inlinenum}
\end{assumption}

Assumption \ref{as: nodewise lasso} imposes regularity conditions in terms of convergence rates on the node-wise Lasso estimators used to approximate the inverse Gram matrix $\bTheta_0 = \E[\bX\bX^\top]^{-1}$. These convergence rates are well-established in the literature on estimation of inverse Gram matrices; see, for example, \cite{G16}. Part \ref{as: max approximation nodewise} requires that the estimators $\widetilde\bTheta_k$ provide sufficiently accurate approximations to the inverse of the empirical Gram matrix computed on the subsample $I(-k)$. Part \ref{as: bounded row nodewise} ensures that each row of the estimator $\widetilde\bTheta_k$ remains bounded in the $\ell_1$ norm. Assumption \ref{as: growth conditions} restricts the growth of the sparsity indices relative to the sample size. In particular, it requires that the sparsity levels $s_{\bgamma}$ and $s_{\bphi}$ do not grow too quickly compared to $n$ and $\log p$. Note also that Assumption \ref{as: growth conditions} is weaker than the corresponding conditions used to derive the $\sqrt n$-consistency and asymptotic normality result for the double Lasso estimator; see \cite{BCH14} for the case of the original double Lasso estimator and \cite{CCDDHNR18} for the case of the double Lasso estimator combined with cross-fitting.

We are now ready to establish the main result of this paper, which shows that the triple Lasso estimator $\widehat\beta^{TL}$ is $\sqrt n$-consistent and asymptotically normal under our assumptions. The proof of this theorem can be found in the Appendix.

\begin{thm}\label{thm: asymptotic normality}
Under Assumptions \ref{as: moments}--\ref{as: growth conditions},
\begin{equation}\label{eq: triple lasso linear approximation}
\sqrt n(\widehat\beta^{TL} - \beta_0)
=
\frac{1}{\sqrt n}\sum_{i=1}^n
\frac{\varepsilon_i\nu_i}{\E[\nu^2]}
+
O_P(R_1)+O_P(R_2)+O_P(R_3),
\end{equation}
where
\begin{align}
R_1 & :=
\frac{s_{\bgamma}^{1/4}\sqrt{s_{\bphi}\log p}}{\sqrt n},
\qquad
R_2 :=
\frac{s_{\bgamma}^{3/2}\sqrt{s_{\bphi}}(\log p)^{3/2}}{n},
\label{eq: r1 and r2}
\\
R_3 &:=
\sqrt{s_{\bphi}\log p}
\left(
\frac{1}{K}\sum_{k=1}^K
\|
\bar\bgamma_0-\bar\bgamma_{0,\widehat T_k}
\|_2
+
\sqrt{\E[|\bX^\top(\bgamma_0-\bar\bgamma_0)|^2]}
+
\|
\bgamma_0-\bar\bgamma_0
\|_2
\right).
\label{eq: r3}
\end{align}
Therefore,
\[
\sqrt n(\widehat\beta^{TL}-\beta_0)
\to_d
N(0,V),\quad V:=\frac{\E[(\varepsilon\nu)^2]}{(\E[\nu^2])^2},
\]
as long as $R_3=o_P(1)$.
\end{thm}


The result in Theorem \ref{thm: asymptotic normality} for the triple Lasso
estimator can be compared with the corresponding result for the double Lasso
estimator. In particular, it follows from the proof of Theorem \ref{thm:
asymptotic normality} that the cross-fitted double Lasso estimator
$\widehat\beta^{DL}$ satisfies
\begin{equation}\label{eq: double lasso linear approximation}
\sqrt n(\widehat\beta^{DL}-\beta_0)
=
\frac{1}{\sqrt n}\sum_{i=1}^n
\frac{\varepsilon_i\nu_i}{\E[\nu^2]}
+
O_P(R_4),
\qquad
R_4:=
\frac{\sqrt{s_{\bgamma}s_{\bphi}}\log p}{\sqrt n},
\end{equation}
which is essentially a well-known result; e.g., see \cite{CCDDHNR18}. Both
estimators are therefore $\sqrt n$-consistent and properly centered
asymptotically normal if the corresponding remainder terms in asymptotic linear
representations \eqref{eq: triple lasso linear approximation} and \eqref{eq:
double lasso linear approximation} are asymptotically vanishing.

In turn, to compare the remainder terms, suppose first that $R_4=o(1)$ so that the double Lasso estimator is $\sqrt n$-consistent and asymptotically normal. We claim that then the remainder terms $R_1$, $R_2$, and $R_3$ for the triple Lasso estimator are {\em always at most of the same order} as the remainder term $R_4$ for the double Lasso estimator and {\em often of smaller order}. Indeed, 
$$
\frac{R_1}{R_4} = \frac{1}{s_{\bgamma}^{1/4}\sqrt{\log p}} = o(1)
$$
whenever either $p\to\infty$ or $s_{\bgamma}\to\infty$, and
$$
\frac{R_2}{R_4} = \frac{s_{\bgamma}\sqrt{\log p}}{\sqrt n} \leq R_4 = o(1)
$$
since $s_{\bphi}\ge s_{\bgamma}$ by the definition of $s_{\bphi}$. Also, in the {\em worst} case, the remainder term $R_3$ is of the same order as $R_4$:
\[
|R_3|
\le
\sqrt{s_{\bphi}\log p}
\left(
\frac{1}{K}\sum_{k=1}^K
\|
\bar\bgamma_0-\widehat\bgamma_k
\|_2
+
\sqrt{\E[|\bX^\top(\bgamma_0-\bar\bgamma_0)|^2]}
+
\|
\bgamma_0-\bar\bgamma_0
\|_2
\right)
\lesssim_P
R_4,
\]
since $\widehat T_k$ contains the support of $\widehat\bgamma_k$, but {\em often} converges to zero faster than $R_4$. Indeed, the latter occurs whenever the approximation errors $\E[|\bX^\top(\bgamma_0-\bar\bgamma_0)|^2]$ and $\|\bgamma_0-\bar\bgamma_0\|_2^2$
vanish faster than the upper bound $s_{\bgamma}/n$ imposed in Assumption \ref{as: sparsity}.\ref{as: sparse approximations} and when most coefficients in the approximating vector $\bar\bgamma_0$ are not approaching zero too quickly. In the extreme case where $\bgamma_0$ is exactly sparse, so that $\bgamma_0-\bar\bgamma_0=\bzero_p$, and all nonzero coefficients in $\bgamma_0$ are sufficiently separated away from zero, so that the Lasso estimators $\widehat\bgamma_k$ perform consistent screening, the remainder term $R_3$ is actually {\em equal} to zero. This explains why the remainder terms in the asymptotic linear representation for the triple Lasso estimator often vanish faster than that of the double Lasso estimator, leading to more accurate inference. In Section \ref{sec: simulations}, we illustrate this phenomenon in Monte Carlo simulations, where we demonstrate that the triple Lasso estimator yields confidence intervals with a better coverage control. In fact, the improvement in coverage control is rather dramatic in some cases.

In addition, we note that the remainder term $R_3$ can be further reduced by enlarging the sets $\widehat T_k$. For example, in practice one can arrange the entries of
\[
S_k :=
\widetilde\bTheta_k
\E_{I(-k)}[\bX_i(D_i-\bX_i^\top\widehat\bgamma_k)]
\]
in decreasing order of absolute value and include the indices of, say, the $L$ largest entries in $\widehat T_{k,1}$, provided that they are not already contained in $\widehat T_{k,0}$. This strategy is motivated by Lemma \ref{lem: auxiliary inequalities} and Assumption \ref{as: nodewise lasso}, which imply that the vectors $S_k$ approximate the corresponding vectors $\bgamma_0-\widehat\bgamma_k$. Allowing $L$ to grow beyond what is permitted by Assumption \ref{as: sparsity}.\ref{as: sparse choice} would reduce $R_3$ at the expense of increasing $R_1$ and $R_2$, which reflects a familiar bias--variance trade-off. In this paper, however, we simply set $L=0$ and leave the question of how to choose $L$ so as to balance these effects for future work.

Finally, in the exactly sparse case with non-zero coefficients of $\bgamma_0$ being sufficiently separated away from zero, so that $R_3=0$, it is actually possible that the triple Lasso estimator remains asymptotically normal even when the double Lasso estimator fails to be asymptotically normal. This can occur when $s_{\bgamma}$ grows sufficiently slowly relative to $s_{\bphi}$, and
$$
\frac{s_{\bgamma}s_{\bphi}(\log p)^2}{n} \to \infty\quad\text{but}\quad \frac{\sqrt{s_{\bgamma}}s_{\bphi}\log p}{n}\to 0\quad\text{and}\quad \frac{s_{\bgamma}^{3/2}\sqrt{s_{\bphi}}(\log p)^{3/2}}{n}\to 0
$$
since in this case we have $R_1\to0$ and $R_2\to0$ but $R_4\to\infty$ as $n\to\infty$.

\section{Extensions}\label{sec: extensions}
In this section, we provide a general recursive formula for constructing a moment function satisfying the Neyman orthogonality condition to any order in a Z-estimation problem. We then apply the general formula to derive the second-order Neyman orthogonal moment function in an M-estimation problem with a single-index structure.
\subsection{General Formula}
Consider a $Z$-estimation problem
\begin{equation}\label{eq: z estimation problem}
\begin{cases}
\E[\bar{f}(\bZ,\beta_{0},\btheta_{0})]=0,\\
\E[\bar{u}(\bZ,\beta_{0},\btheta_{0})]=\bzero_p,
\end{cases}
\end{equation}
where $\bZ$ is an observable random vector, $\beta_0\in\R$ is the parameter of interest, $\btheta_0\in\R^p$ is a vector of nuisance parameters, and $(\bar f,\bar u)$ is a pair of known functions, with $\bar u$ being $\R^p$-valued. In this subsection, we describe a method to construct a $k$th order Neyman orthogonal moment function for estimating $\beta_0$, for any integer $k\geq 1$. 

To construct our moment function, we proceed recursively. Let $\bZ^{(1)},\bZ^{(2)},\dots$ be independent copies of $\bZ$ and denote $\bZ^k = (\bZ^{(1)},\dots,\bZ^{(k)})$ if $k\geq 1$ and $\bZ^k = \bZ$ if $k=0$. Fix any $k\geq 1$ and suppose that we have already constructed moment functions $F\colon\R\times\R^q\to\R$ and $U\colon\R\times\R^q\to\R^q$ of the form $F(\beta,\bet) = \E[f(\bZ^{k-1},\beta,\bet)]$ and $U(\beta,\bet) = \E[u(\bZ^{k-1},\beta,\bet)]$ such that the pair $(\beta_0,\bet_0)$ solves the system of moment equations
$$
\begin{cases}
F(\beta,\bet) = 0,\\
U(\beta,\bet)=\bzero_q,
\end{cases}
$$
where $\bet_0 = (\eta_{0,1},\dots,\eta_{0,q})^\top \in\R^q$ is a nuisance parameter, and $F$ satisfies the $(k-1)$th order Neyman orthogonality condition:
$$
\nabla_{\bet}^{m}F(\beta_0,\bet_0) = \bzero_q^{\otimes m},\quad\text{for all }m\in[k-1],
$$ 
where $\nabla_{\bet}^m F(\beta_0,\bet_0)$ is a tensor in $(\R^q)^{\otimes m}:=\R^{q\times \dots \times q}$ defined by
$$
(\nabla_{\bet}^mF(\beta_0,\bet_0))_{j_1,\dots, j_m} = \frac{\partial^mF(\beta_0,\bet_0)}{\partial\eta_{j_1}\dots,\partial\eta_{j_m}},\quad\text{for all }j_1,\dots,j_m \in [q], 
$$ 
and for any vector $\ba = (a_1,\dots,a_q)^\top$ in $\R^q$, we use $\ba^{\otimes m}$ to denote the tensor in $(R^q)^{\otimes m}$ defined by 
$$
(\ba^{\otimes m})_{j_1,\dots, j_m} = a_{j_1}\dots a_{j_m},\quad \text{for all }j_1,\dots,j_m \in [q].
$$
Now, let 
\begin{equation}\label{eq: a0 and b0 definition}
\bbA_0:=\frac{\partial U(\beta_0,\bet_0)^\top}{\partial\bet}\quad\text{and}\quad \bbB_0:= (\bbA_0^{-1})^{\otimes k}\nabla_{\bet}^kF(\beta_0,\bet_0),
\end{equation}
where the latter includes the mode-wise tensor-matrix product notation: for any matrix $\mathbf C = (C_{j_1,j_2})_{j_1,j_2=1}^q$ in $\R^{q\times q}$ and any tensor $\mathbf D = (D_{j_1,\dots,j_m})_{j_1,\dots,j_m=1}^q$ in $(\R^q)^{\otimes m}$, the product $\mathbf C^{\otimes m}\mathbf D$ is the tensor in $(\R^q)^{\otimes m}$ defined by
$$
(\mathbf C^{\otimes m}\mathbf D)_{j_1,\dots,j_m} = \sum_{i_1,\dots,i_m=1}^q C_{j_1,i_1}\dots C_{j_m, i_m}D_{i_1,\dots,i_m},\quad\text{for all }j_1,\dots,j_m \in [q].
$$
Then define the moment functions $\tilde F\colon \R\times \R^{q + q^2 + q^k}\to\R$ and $\tilde U\colon \R\times \R^{q + q^2 + q^k}\to\R^{q + q^2 + q^k}$ by setting 
$$
\tilde F(\beta,\tilde\bet) := F(\beta,\bet) - \frac{1}{k!} \bbB[(U(\beta,\bet))^{\otimes k}],
$$
and
$$
\tilde U(\beta,\tilde\bet) := 
\left(\begin{array}{c}
U(\beta,\bet)\\
\vec(\partial U(\beta,\bet)^\top/\partial\bet - \bbA)\\
\vec(\bbA^{\otimes k}\bbB - \nabla^k_{\bet}F(\beta,\bet))
\end{array}\right),
$$
where $\tilde\bet = (\bet^\top,\vec(\bbA)^\top,\vec(\bbB)^\top)^\top$ and we used the contraction of tensors notation: for any tensors $\mathbf D = (D_{j_1,\dots,j_m})_{j_1,\dots,j_m=1}^q$ and $\mathbf E = (E_{j_1,\dots,j_m})_{j_1,\dots,j_m=1}^q$ in $(\R^q)^{\otimes m}$, the contraction $\mathbf D[\mathbf E]$ is a scalar in $\R$ defined by
$$
\mathbf D[\mathbf E] = \sum_{j_1,\dots,j_m=1}^q D_{j_1,\dots,j_m}E_{j_1,\dots,j_m}.
$$
Note in passing that $\tilde F$ and $\tilde U$ take the form $\tilde F(\beta,\tilde\bet) = \E[\tilde f(\bZ^{k},\beta,\tilde\bet)]$ and $\tilde U(\beta,\tilde \bet) = \E[u(\bZ^{k},\beta,\tilde\bet)]$, maintaining the recursion.
The following lemma shows that $\tilde F$ is the desired moment function. The proof of this lemma can be found in the Appendix.
\begin{lem}\label{lem: higher-order neyman orthogonal score}
As long as the function $\bet\mapsto U(\beta_0,\bet)$ is continuously differentiable, the function $\bet\mapsto F(\beta_0,\bet)$ is $k$ times continuously differentiable, and the matrix $\bbA_0$ is invertible, the pair $(\beta_0,\tilde\bet_0)$ solves the system of moment equations
$$
\begin{cases}
\tilde F(\beta,\tilde\bet) = 0,\\
\tilde U(\beta,\tilde\bet) = \bzero_{q + q^2 + q^k},
\end{cases}
$$
and $\tilde F$ satisfies the $k$th order Neyman orthogonality condition:
$$
\nabla_{\tilde\bet}^m \tilde F(\beta_0,\tilde\bet_0) = \bzero_{q + q^2 + q^k}^{\otimes m},\quad\text{for all }m\in[k],
$$
where $\tilde\bet_0 := (\bet_0^\top,\vec(\bbA_0)^\top,\vec(\bbB_0)^\top)^\top$.
\end{lem}
\begin{rem}
Here, we provide intuition for our formula without relying on tensor notation by considering the case with $k=2$ and assuming that the moment function $F$ already satisfies the first-order Neyman orthogonality condition; see \cite{CCDDHNR18} for the construction of the first-order Neyman orthogonal moment function starting from the original moment functions in \eqref{eq: z estimation problem}.

Intuitively, in order to estimate $\beta_0$, we would like to use the moment function $\beta\mapsto F(\beta,\bet_0)$. However, $\bet_0$ is unknown and we have to replace it by the corresponding estimator $\widehat\bet$, which provides an approximating moment function $\beta\mapsto F(\beta,\widehat\bet)$. The key idea behind the construction of moment functions satisfying Neyman orthogonality conditions is to improve the approximating function $\beta\mapsto F(\beta,\widehat\bet)$ by considering Taylor's expansion of $F(\beta_0,\widehat\bet)$ around $\widehat\bet = \bet_0$, so that
\begin{equation}\label{eq: taylor expansion second order}
F(\beta_0,\bet_0) = F(\beta_0,\widehat\bet) - \frac{1}{2}(\widehat\bet - \bet_0)^\top \frac{\partial^2 F(\beta_0,\bet_0)}{\partial\bet \partial \bet^\top}(\widehat\bet - \bet_0) + o(\|\widehat\bet - \bet_0\|_2^2),
\end{equation}
where we used the fact that $F$ satisfies the first-order Neyman orthogonality condition. The difference of the first two terms on the right-hand side of this equation provides a better approximation to $F(\beta_0,\bet_0)$ than $F(\beta_0,\widehat\bet)$ does itself but one of the problems here is that the difference $\widehat\bet - \bet_0$ is not observed. To solve this problem, we consider Taylor's expansion of $U(\beta_0,\widehat\bet)$ around $\widehat\bet = \bet_0$:
$$
U(\beta_0,\widehat\bet) = \frac{\partial U(\beta_0,\bet_0)}{\partial \bet^\top}(\widehat\bet - \bet_0) + o(\|\widehat\bet - \bet_0\|_2),
$$
where we used the fact that $U(\beta_0,\bet_0) = \bzero_q$. Solving this system of equations for $\widehat\bet - \bet_0$ and substituting the solution to \eqref{eq: taylor expansion second order} in turn yields
$$
F(\beta_0,\bet_0) = F(\beta_0,\widehat\bet) - \frac{1}{2}U(\beta_0,\widehat\bet)^\top \bbB_0 U(\beta_0,\widehat\bet) + o(\|\widehat\bet - \bet_0\|_2^2),
$$
where we denoted
\begin{equation}\label{eq: a0 and b0 simplification}
\bbA_0:=\frac{\partial U(\beta_0,\bet_0)^\top}{\partial \bet}\quad\text{and}\quad \bbB_0 := \bbA_0^{-1} \frac{\partial^2 F(\beta_0,\bet_0)}{\partial\bet \partial \bet^\top} (\bbA_0^\top)^{-1}.
\end{equation}
This yields the moment function 
$$
\tilde F(\beta,(\bet^\top,\vec(\bbA)^\top,\vec(\bbB)^\top)^\top) = F(\beta,\bet) - \frac{1}{2}U(\beta,\bet)^\top \bbB U(\beta,\bet)
$$
that is second-order Neyman orthogonal with respect to $\bet$ by construction. Fortunately, it is also second-order Neyman orthogonal with respect to $\vec(\bbA)$ and $\vec(\bbB)$, and thus fully second-order Neyman orthogonal. The construction provided in Lemma \ref{lem: higher-order neyman orthogonal score} generalizes the idea described above to the $k$th order Neyman orthogonality for any integer $k$ using the tensor notation.
\qed
\end{rem}

\subsection{Application to M-Estimation with Single-Index Structure}
We now apply the general formula obtained in the previous subsection to derive the second-order Neyman orthogonal moment function for estimating $\beta_0$ in the M-estimation problem with a single-index structure
$$
(\beta_0,\btheta_0):=\argmin_{(\beta,\btheta)\in\R\times\R^p}\E[m(D\beta+\bX^\top\btheta,\bY)],
$$
where $m\colon\R\times\mathcal Y\to\R$ is a known smooth loss function, $\bY\in\mathcal Y$ is one or more outcome variables, $D\in\R$ is a regressor of interest, $\bX\in\R^p$ is a vector of controls, $\beta_0\in\R$ is the parameter of interest, and $\btheta_0\in\R^p$ is a vector of nuisance parameters.

To this end, denote the first, second, and third derivatives of $m$ with respect to its first argument by $m_1'$, $m_{11}''$, and $m_{111}'''$, respectively, and let $\bmu_0\in\R^p$ be a solution to the moment equation
$$
\E[m_{11}''(D\beta_0 + \bX^\top\theta_0,\bY)(D-\bX^\top\bmu)\bX] = \bzero_p.
$$
Then the first-order Neyman orthogonal moment function for estimating $\beta_0$ is $F\colon\R\times\R^{2p}\to\R$ given by
$$
F(\beta,\bet):=\E[m_1'(D\beta + \bX^\top\btheta,\bY)(D-\bX^\top\bmu)],\quad \bet:=(\btheta^\top,\bmu^\top)^\top;
$$
e.g. see \cite{GBRD14} or \cite{chetverikov2025selecting}. Thus, setting
$$
U(\beta,\bet):=\left(\begin{array}{c}
\E[m_1'(D\beta + \bX^\top\btheta,\bY)\bX]\\
\E[m_{11}''(D\beta + \bX^\top\btheta,\bY)(D-\bX^\top\bmu)\bX]
\end{array}\right),
$$
we fit our problem into a setting of the previous subsection with $\bet_0:= (\btheta_0^\top,\bmu_0^\top)^\top$.

Now, to apply the general formula from the previous subsection, let
\begin{align*}
& \mathbf G_0 := \E[m_{11}''(D\beta_0 + \bX^\top\btheta_0,\bY)\bX\bX^\top],\\
& \mathbf H_0 := \E[m_{111}'''(D\beta_0 + \bX^\top\btheta_0,\bY)(D-\bX^\top\bmu_0)\bX\bX^\top].
\end{align*}
Then
$$
\frac{\partial U(\beta_0,\bet_0)}{\partial\bet^\top}
=
\begin{pmatrix}
\mathbf G_0 & \bzero\\
\mathbf H_0 & -\mathbf G_0
\end{pmatrix},
\qquad
\nabla_{\bet}^2 F(\beta_0,\bet_0)
=
\begin{pmatrix}
\mathbf H_0 & -\mathbf G_0\\
-\mathbf G_0 & \bzero
\end{pmatrix}.
$$
Hence, with $\bbA_0$ and $\bbB_0$ defined by \eqref{eq: a0 and b0 definition}, which simplifies to \eqref{eq: a0 and b0 simplification} for $k=2$, we have
\[
\bbB_0
=
\begin{pmatrix}
-\,\mathbf G_0^{-1}\mathbf H_0\mathbf G_0^{-1} & \mathbf G_0^{-1}\\[1mm]
\mathbf G_0^{-1} & \bzero
\end{pmatrix}.
\]
Therefore, the second-order Neyman orthogonal moment function is
\begin{align*}
\tilde F(\beta,\tilde\bet)
&=
\E[m_1'(D\beta + \bX^\top\btheta,\bY)(D-\bX^\top\bmu)] \\
&\quad
-\E[m_1'(D\beta + \bX^\top\btheta,\bY)\bX]^\top
\mathbf G^{-1}
\E[m_{11}''(D\beta + \bX^\top\btheta,\bY)(D-\bX^\top\bmu)\bX] \\
&\quad
+\frac12\,
\E[m_1'(D\beta + \bX^\top\btheta,\bY)\bX]^\top
\mathbf G^{-1}\mathbf H\mathbf G^{-1}
\E[m_1'(D\beta + \bX^\top\btheta,\bY)\bX],
\end{align*}
where $\tilde\bet := (\btheta^\top,\bmu^\top,\vec(\mathbf G)^\top,\vec(\mathbf H)^\top)^\top$ and $\tilde\bet_0 := (\btheta_0^\top,\bmu_0^\top,\vec(\mathbf G_0)^\top,\vec(\mathbf H_0)^\top)^\top$. It is also straightforward to verify that $\beta_0$ solves the desired moment equation $\tilde F(\beta,\tilde\bet_0) = 0$.

\begin{rem}
It is useful to note that the moment function derived above reduces to the moment function $\psi^{TL}$ in Section \ref{sec: triple lasso} if we set $m(\cdot,\bY) = (\bY - \cdot)^2$ to obtain the model from Section \ref{sec: triple lasso}. Note also that one can construct an estimator based on this moment function analogous to the triple Lasso estimator in Section \ref{sec: triple lasso} but we leave such a development to future work.
\qed
\end{rem}

\section{Monte Carlo Simulations}\label{sec: simulations}

In this section, we study the finite-sample behavior of the triple Lasso estimator
relative to the standard cross-fitted double Lasso estimator
in the linear regression model from Section \ref{sec: triple lasso}. In our Monte Carlo exercises, we vary the
sparsity of the nuisance vector $\bgamma_0$, the correlation structure of the
controls $\bX$, and the sample/problem size, while keeping the target parameter
fixed at $\beta_0 = 1$.

\subsection{Design}

For each Monte Carlo replication, we generate i.i.d.\ observations
$(\bX_i,D_i,Y_i)$, $i=1,\dots,n$, from model \eqref{eq:
model}--\eqref{eq:first-stage} with $\beta_0 = 1$, $\nu\sim\cN(0,1)$ and
$\varepsilon\sim\cN(0,1)$ independent of each other and of $\bX$. The controls
follow a Gaussian design,
\[
\bX \sim \cN(\bzero,\bSigma_0),
\qquad \bSigma_0:=\bSigma_0(\rho):=\sbr[1]{\rho^{|j-k|}}_{j,k=1}^p,
\]
so that the scalar $\rho \in \{0,0.2,0.4,0.6,0.8\}$ governs the correlation
between controls. Note that with $\bX\sim \cN(\bzero,\bSigma_0)$ and $\bSigma_0$
of the Toeplitz form given above, the inverse $\bOmega_0=\bSigma_0^{-1}$ takes a
tri-diagonal form. Joint normality therefore leads to any single control $X_j$
being conditionally normally distributed given the other controls $\bX_{-j}$.
Unpacking $\bOmega_0$, we get
\[
\left.X_j \middle| \bX_{-j}=\bx_{-j}\right.\sim
\begin{cases}
\cN(\rho x_2, 1-\rho^2), & j=1, \\
\cN\left(\dfrac{\rho}{1+\rho^2}(x_{j-1}+x_{j+1}),
\dfrac{1-\rho^2}{1+\rho^2}\right), & 2 \le j \le p-1, \\
\cN(\rho x_{p-1}, 1-\rho^2), & j=p.
\end{cases}
\]
In particular, the conditional means are sparse linear functions of the
conditioning variables.

We consider the sample sizes $n \in \{500,1000,2000\}$ and limit attention to
the high-dimensional regime $p=n/2$. The structure of the outcome nuisance
vector $\btheta_0$ is held fixed across designs and is approximately sparse with
geometrically decaying entries,
\[
\theta_{0,j} = (0.5)^{\,j-1}, \quad j=1,\ldots,p.
\]
The treatment nuisance vector $\bgamma_0$ also has (at least) geometrically
decaying entries
\[
\gamma_{0,j} = 
\begin{cases}
  (0.5)^{\,j-1}, & j \leq s_{\bgamma},\\
  0, & j > s_{\bgamma},
\end{cases}
\]
but its sparsity varies across three regimes:
\begin{enumerate}[label=(\roman*)]
\item \emph{Exact sparsity:} $s_{\bgamma} = 2$;
\item \emph{Intermediate sparsity:} $s_{\bgamma} = 5$; and,
\item \emph{Approximate sparsity:} $s_{\bgamma} = p$.
\end{enumerate}
For each of the $3\times5\times3=45$ design points $(n,\rho,s_{\bgamma})$, we run
$2{,}000$ Monte Carlo replications.

The approximate sparsity of $\btheta_0$ is already challenging for the double
lasso. Moving from exact to approximate sparsity of $\bgamma_0$ gradually
weakens the environment in which even our second-order orthogonal procedure is
expected to perform well since approximate sparsity design implies larger effective sparsity index $s_{\bgamma}$.

\subsection{Estimation and Implementation}

We compare two estimators. First, the double Lasso estimator is implemented as the
cross-fitted residual-on-residual estimator using the DML1 aggregation rule of
\cite{CCDDHNR18}. In each fold, we estimate the treatment regression $D$ on
$\bX$ and the reduced-form regression $Y$ on $\bX$ by lasso on the training
sample, evaluate the fitted values on the holdout sample, and form the
fold-specific orthogonal score based on the residuals.

Second, the triple Lasso estimator is implemented as the cross-fitted adjusted-score
estimator. Relative to the double Lasso estimator, it augments the holdout score by the
estimated second-order correction term built from estimates of selected rows of
the inverse control correlation matrix. Operationally, after estimating
$\widehat{\bgamma}_k$ and $\widehat{\bphi}_k$ on the training fold, we define
the selected support
\(
\widehat T_k = \{j : \widehat\gamma_{k,j} \neq 0\}
\),
estimate only the rows of $\bTheta_0 = \bSigma_0^{-1}$ indexed by $\widehat T_k$
using node-wise Lasso, and then evaluate the adjusted score on the holdout
observations. We then use DML1 aggregation to produce the triple lasso estimate.

Both procedures use the same random $K=5$ folds in each replication. All Lasso
penalties are determined by plug-in rules of the Bickel--Ritov--Tsybakov type
\citep{bickel_simultaneous_2009}. Specifically, let $n_K := (K-1)n/K$ denote the
training-sample size and introduce the Belloni--Chen--Chernozhukov--Hansen
baseline penalty level
\begin{equation}\label{eq:bcch_penalty_for_glmnet_lasso}
\lambda(n_{\text{obs}},p_{\text{pen}})
:=
\frac{c}{\sqrt n_{\text{obs}}}\Phi^{-1}\left(1 - \frac{\alpha(n_{\text{obs}},p_{\text{pen}})}{2p_{\text{pen}}}\right),
\end{equation}
where $c:=1.1$ and
\(
\alpha(n_{\text{obs}},p_{\text{pen}}):= 0.1/\log(\max\{p_{\text{pen}},n_{\text{obs}}\}),
\)
as suggested by \citet{belloni_sparse_2012}, with $n_{\text{obs}}$ and
$p_{\text{pen}}$ being placeholders for the numbers of observations and
penalized parameters, respectively. The fold-invariant penalty levels are then
set as
\begin{equation}\label{eq:brt_penalties_simulations}
\lambda_\gamma := \lambda(n_K,p)\sigma_\nu,
\qquad
\lambda_\phi := \lambda(n_K,p)\sigma_e,
\qquad
\lambda_{\psi_j}(\rho) := \lambda(n_K,p-1)\sigma_{X_j|\bX_{-j}}(\rho),
\end{equation}
where
\(
\sigma^2_e := \beta_0^2 \sigma_\nu^2 + \sigma_\varepsilon^2 = 2
\)
is the variance of the reduced-form error 
\(
e_i = \beta_0 \nu_i + \varepsilon_i,
\)
and 
\[
  \sigma^2_{X_j\mid \bX_{-j}}(\rho):=\frac{1-\rho^2}{1+\bone\{1<j<p\}\rho^2}
\]
is the conditional variance of control $X_j$ given all other controls
$\bX_{-j}$. Here, $\lambda_\gamma$ is used in the treatment Lasso, $\lambda_\phi$
in the outcome reduced-form Lasso, and the $\{\lambda_{\psi_j}(\rho)\}_{j=1}^p$
are used in the node-wise Lasso step of the triple Lasso estimator.\footnote{ The
(infeasible) penalties in \eqref{eq:brt_penalties_simulations} were chosen for
computational simplicity. A more realistic comparison would use tuning parameter
choices which are feasible in practice. One option is to replace the unknown
standard deviations with (initally conservative) proxies to produce feasible
analogs of the \citet{bickel_simultaneous_2009} penalties. One could refine such
proxies in a possibly iterative manner borrowing ideas from \citet[Algorithm
A.1]{belloni_sparse_2012}. Alternatively, one could use the bootstrap after
cross-validation method proposed in \citet{chetverikov2025selecting}, which
takes into account the correlation structure.}

All Lasso estimators are fit using \texttt{glmnet} in \texttt{R} with standardized
controls.\footnote{The \texttt{glmnet} definition of Lasso is based on \emph{one
half} square loss, which cancels out a ``2'' in the original
\citet{belloni_sparse_2012} baseline penalty, thus leading to our
\eqref{eq:bcch_penalty_for_glmnet_lasso}.} The treatment and reduced-form Lasso
regressions include intercepts; node-wise regressions are run without intercepts.
The reported standard errors are heteroskedasticity-robust score-based standard
errors computed from the corresponding cross-fitted influence function.

\subsection{Performance Metrics}

For each estimator $\widehat\beta$ we report the Monte Carlo
\begin{enumerate}[label=(\roman*)]
\item squared bias, $(\E[\widehat\beta]-\beta_0)^2$;
\item variance, $\var(\widehat\beta)$;
\item mean squared error, $\E[(\widehat\beta-\beta_0)^2]$;
\item coverage of the nominal $95\%$ confidence interval
\(
[\widehat\beta \pm 1.96 \cdot \widehat{\mathrm{se}}(\widehat\beta)];
\)
\item mean confidence-interval length, $2 \cdot 1.96 \cdot
\widehat{\mathrm{se}}(\widehat\beta)$; and, finally,
\item the Monte Carlo distribution of the studentized statistic
\(
(\widehat\beta-\beta_0)/\widehat{\mathrm{se}}(\widehat\beta)
\)
to be contrasted with the $\cN(0,1)$ distribution.
\end{enumerate}
Error bars indicate $\pm 1.96$ Monte Carlo standard errors. To facilitate
comparison, for the studentized statistics, we plot the kernel densities instead
of histograms.\footnote{All kernel densities are created using the
\texttt{ggplot2} with \texttt{geom\_density}. In expectation of an approximately
normal distribution, we use a Gaussian kernel and the \citet[Equation
(3.31)]{silverman_density_1986} rule-of-thumb bandwidth (both
\texttt{geom\_density }defaults).}

\subsection{Results}

Figures \ref{fig:sim-bias}--\ref{fig:sim-density-rho0} show the performance of
the triple and double Lasso estimators based on the stated performance metrics in turn.

A central pattern from the figures is that triple Lasso estimator substantially reduces
bias precisely in the designs where one would expect second-order effects to
matter most. This is most visible in Figure \ref{fig:sim-bias}, which plots the \emph{squared bias}.
For example, under approximate sparsity with $(n,p,\rho)=(1000,500,0)$, the
squared bias falls from $0.00216$ for the double Lasso estimator to $0.00025$ for the triple
Lasso estimator; with $(n, p,\rho)=(2000,1000,0)$, it falls from $0.00082$ to $0.00009$.
Similar reductions appear throughout the intermediate-sparsity designs. Note that these findings do not contradict the asymptotic theory in Section \ref{sec: triple lasso}, which suggests that the gains of the triple Lasso estimator are largest in the exactly sparse designs with consistent screening. Indeed, our approximately sparse design corresponds to larger effective sparsity index relative to our exact sparsity design, which hurts the double Lasso estimator more than it hurts the triple Lasso estimator by construction.

The \emph{variance} of the triple Lasso estimator is typically somewhat larger (Figure
\ref{fig:sim-variance}), which is the expected cost of estimating the additional
score adjustment, but the increase is typically modest relative to the bias
reduction in the difficult designs.

Figure \ref{fig:sim-mse} shows that these bias gains translate into meaningful
improvements in \emph{mean squared error} whenever the double Lasso estimator is materially
biased. The gains are especially pronounced at low and moderate $\rho$ values and under
intermediate or approximate sparsity. For instance, with
$(n,p,\rho,s_{\bgamma})=(500,250,0,\text{interm.})$, the MSE declines from
$0.00630$ to $0.00275$, and with $(1000,500,0,\text{approx.})$ it declines
from $0.00301$ to $0.00129$. As $n$ grows, the gap narrows, and in the easiest
designs, where $\bgamma_0$ is exactly sparse and regressor correlation is high,
the two methods become similar and the double Lasso estimator can occasionally have slightly
smaller MSE because of its lower variance.

Figure \ref{fig:sim-cover} shows the empirical \emph{coverage}. The double Lasso estimator
exhibits substantial undercoverage in the harder designs, again reflecting
residual bias. For example, at $(500,250,0,\text{approx.})$, its empirical $95\%$
coverage is $0.630$; at $(1000,500,0,\text{approx.})$ it is $0.672$; and at
$(2000,1000,0,\text{approx.})$ it is still only $0.750$. The triple Lasso estimator moves
coverage much closer to the nominal level in those same cells, to $0.909$,
$0.923$, and $0.930$, respectively. Even when the improvement in MSE becomes
small, the coverage results indicate that the second-order adjustment remains
useful for inference.

Figure \ref{fig:sim-ci} shows the corresponding tradeoff in \emph{interval
length}. The triple Lasso confidence intervals are systematically longer than double
lasso intervals, which is consistent with the modest increase in variance. The
longer intervals should not be viewed as a drawback per se. Rather, they reflect
the additional uncertainty captured by the second-order adjustment and are
therefore the means by which coverage is brought closer to its nominal level in
designs where the first-order double Lasso procedure understates uncertainty.

Finally, Figure \ref{fig:sim-density-rho0} examines the \emph{studentized
statistics} in the low-correlation designs $(\rho=0)$. The density of the double
Lasso $t$-statistic is visibly shifted away from the standard normal benchmark,
whereas the triple Lasso statistic appears more nearly centered and generally
more aligned with the $\cN(0,1)$ reference curve. The underlying Monte Carlo
moments tell a similar story: for example, under approximate sparsity with
$(n,p,\rho)=(1000,500,0)$, the mean of the double Lasso studentized statistic is
$1.53$, while the mean for triple Lasso is $0.50$. This pattern is consistent
with our second-order orthogonality result.

In Appendix \ref{sec:additional-simulations}, we report the corresponding
densities for $\rho\in\{0.2,0.4,0.6,0.8\}$. Across many of these designs, the
triple Lasso studentized statistic also appears better centered and often closer
to the $\cN(0,1)$ benchmark than the double Lasso statistic, although the
magnitude of the difference varies across cells.

Overall, the simulations support the main theoretical message of the paper. When
nuisance estimation errors are sufficiently small, triple and double Lasso estimators
behave similarly. When second-order terms are non-negligible, however, the triple
Lasso estimator delivers a clear reduction in bias and a substantial improvement in
inference.

\begin{singlespace}
  \bibliographystyle{ecta}
  \bibliography{lasso}   
\end{singlespace}

\begin{figure}[H]
\caption{Monte Carlo squared bias. Rows correspond to the three sparsity regimes
for $\bgamma_0$, and columns correspond to
$(n,p)\in\{(500,250),(1000,500),(2000,1000)\}$. Triple lasso sharply lowers
squared bias in the intermediate and approximate sparsity designs.}
\centering
\includegraphics[width=\textwidth,height=0.88\textheight,keepaspectratio,trim=0 0 0 50,clip]{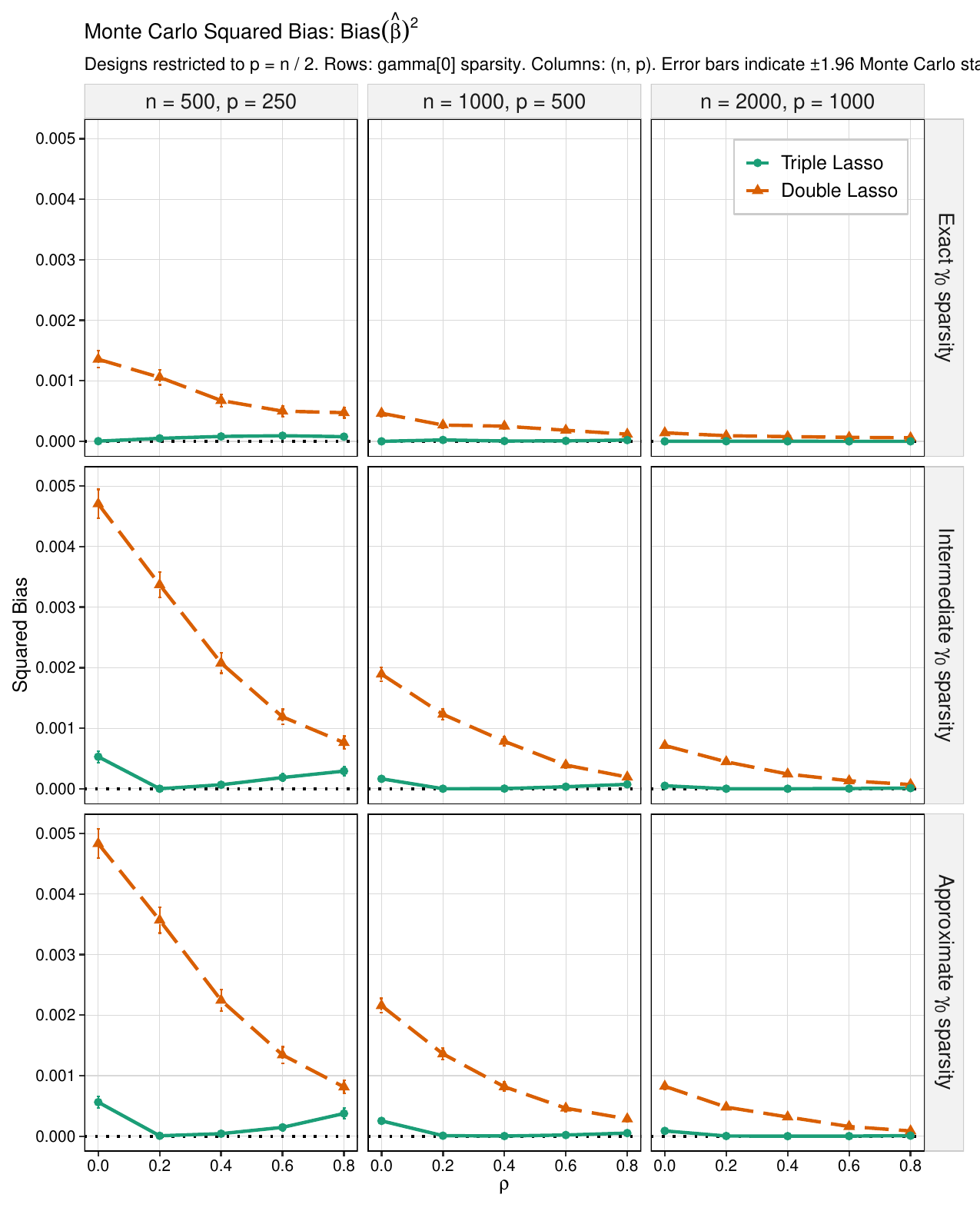}
\label{fig:sim-bias}
\end{figure}

\begin{figure}[H]
\caption{Monte Carlo variance. Rows correspond to the three sparsity regimes for
$\bgamma_0$, and columns correspond to
$(n,p)\in\{(500,250),(1000,500),(2000,1000)\}$. Triple lasso typically incurs a
moderate variance increase relative to double lasso.}
\centering
\includegraphics[width=\textwidth,height=0.88\textheight,keepaspectratio,trim=0 0 0 50,clip]{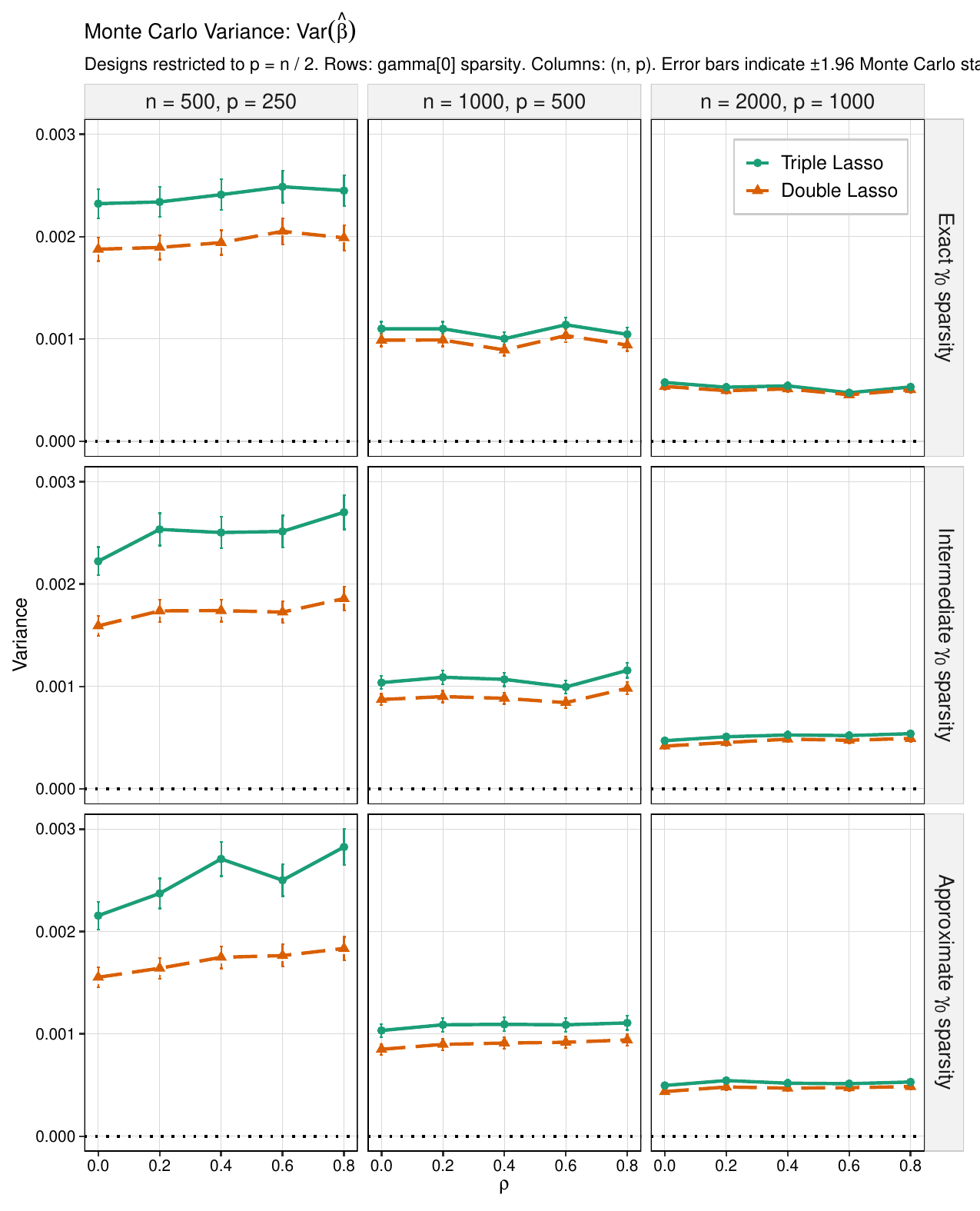}
\label{fig:sim-variance}
\end{figure}

\begin{figure}[H]
\caption{Monte Carlo mean squared error. Triple lasso typically improves MSE in the harder designs, especially under intermediate and approximate sparsity.}
\centering
\includegraphics[width=\textwidth,height=0.88\textheight,keepaspectratio,trim=0 0 0 50,clip]{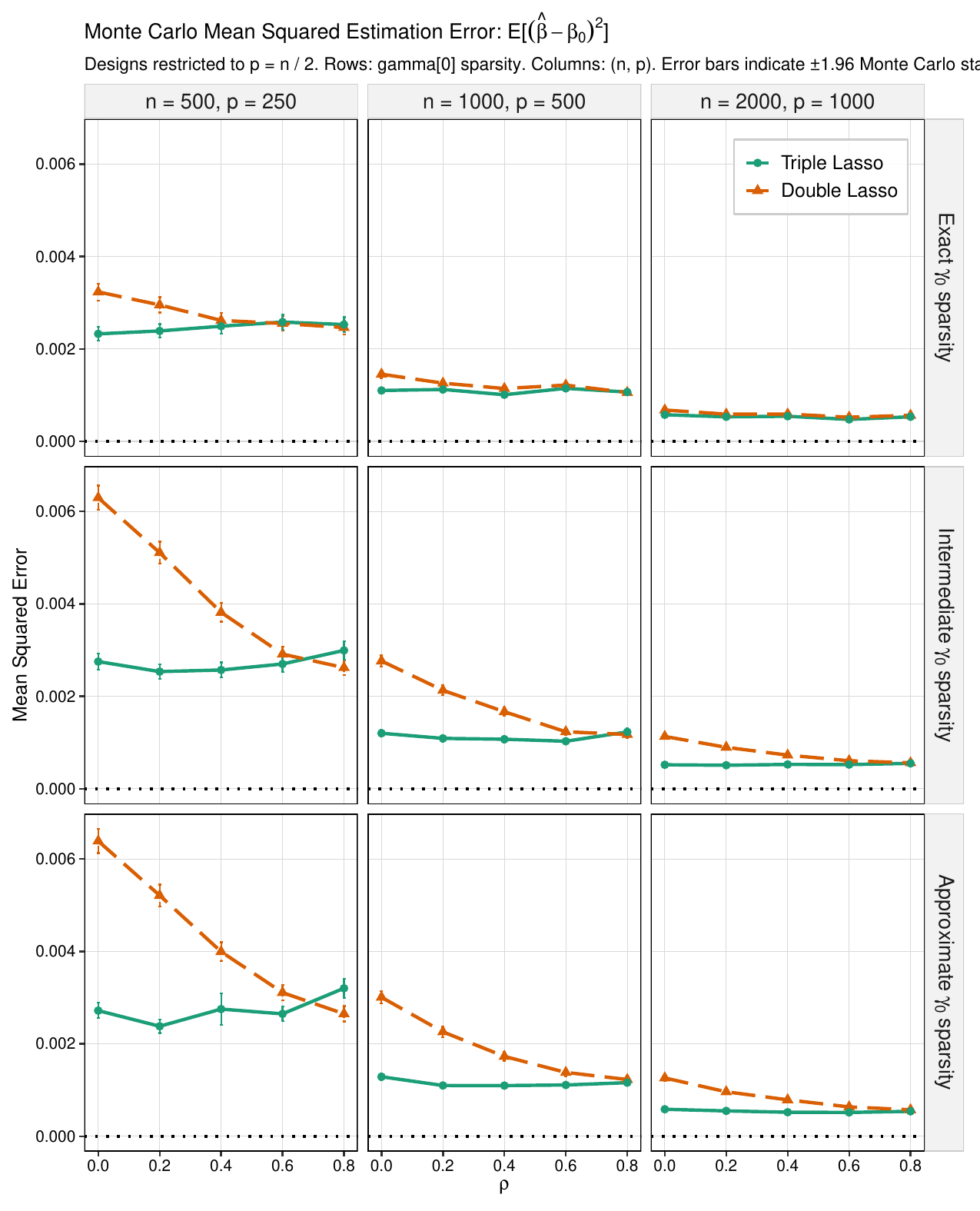}
\label{fig:sim-mse}
\end{figure}

\begin{figure}[H]
\caption{Monte Carlo coverage probability for nominal $95\%$ confidence intervals. Triple lasso brings coverage much closer to the nominal level when double lasso undercovers because of residual bias.}
\centering
\includegraphics[width=\textwidth,height=0.88\textheight,keepaspectratio,trim=0 0 0 50,clip]{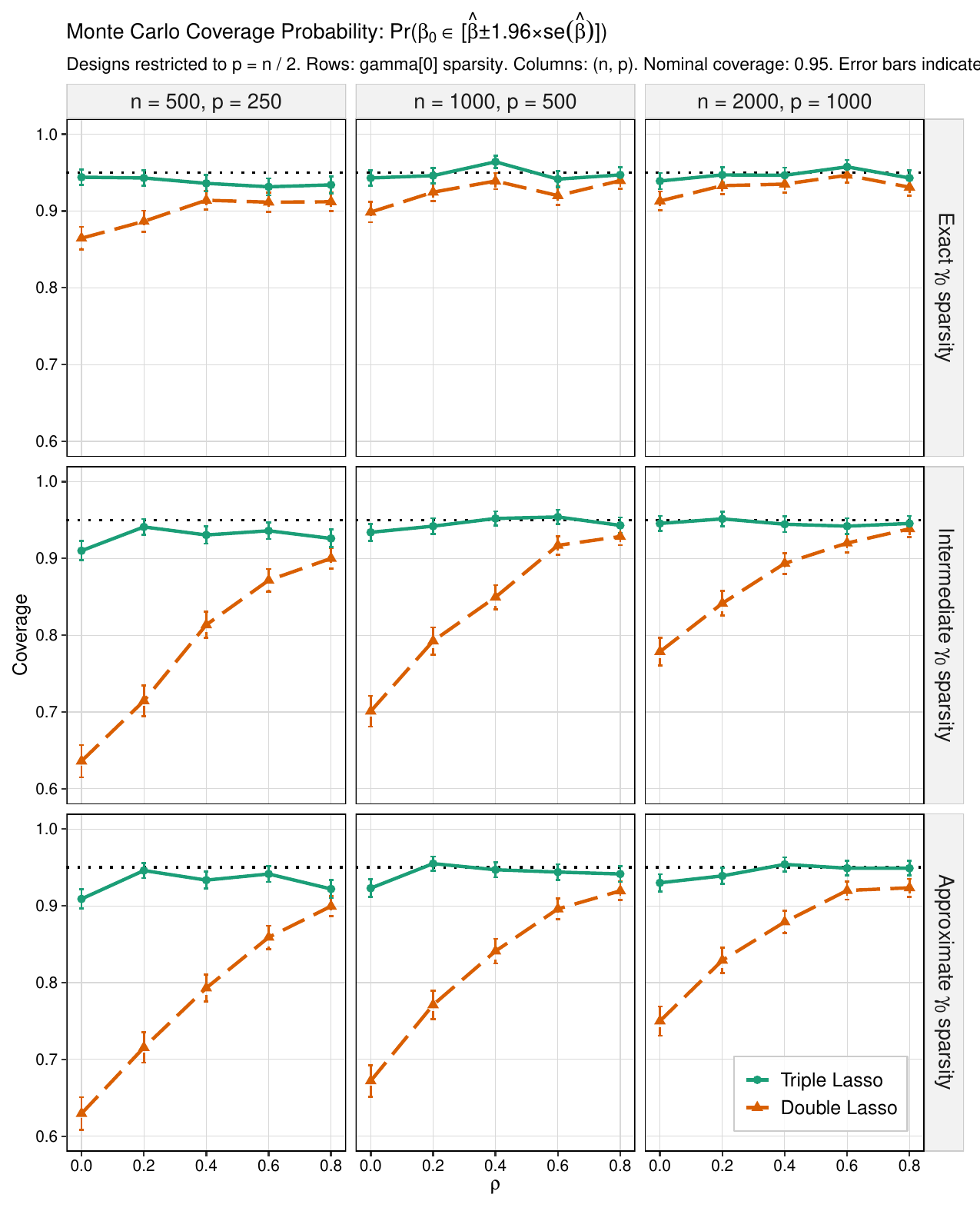}
\label{fig:sim-cover}
\end{figure}

\begin{figure}[H]
\caption{Monte Carlo mean confidence interval length. Triple lasso intervals are somewhat longer, reflecting the additional uncertainty from the second-order correction, but the increase in length is accompanied by substantially more reliable coverage in difficult designs.}
\centering
\includegraphics[width=\textwidth,height=0.88\textheight,keepaspectratio,trim=0 0 0 50,clip]{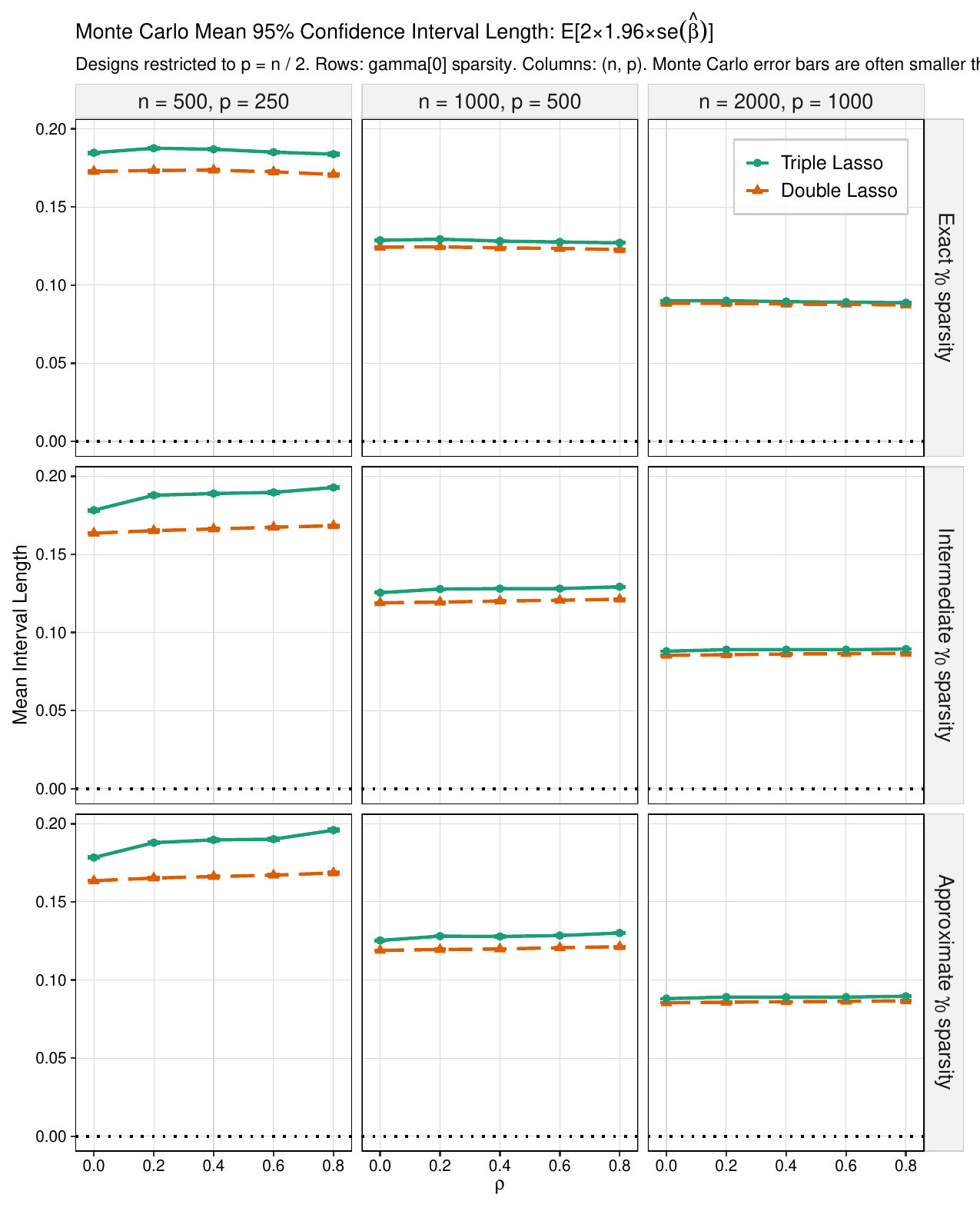}
\label{fig:sim-ci}
\end{figure}

\begin{figure}[H]
\caption{Kernel densities of studentized estimates for $\rho=0$. The dotted line
is the standard normal density. Triple lasso is visibly better centered in the
less sparse designs and closer to standard normal than double lasso for all
designs. These findings are consistent with the improved coverage in Figure
\ref{fig:sim-cover}.}
\centering
\includegraphics[width=\textwidth,height=0.88\textheight,keepaspectratio,trim=0 0 0 40,clip]{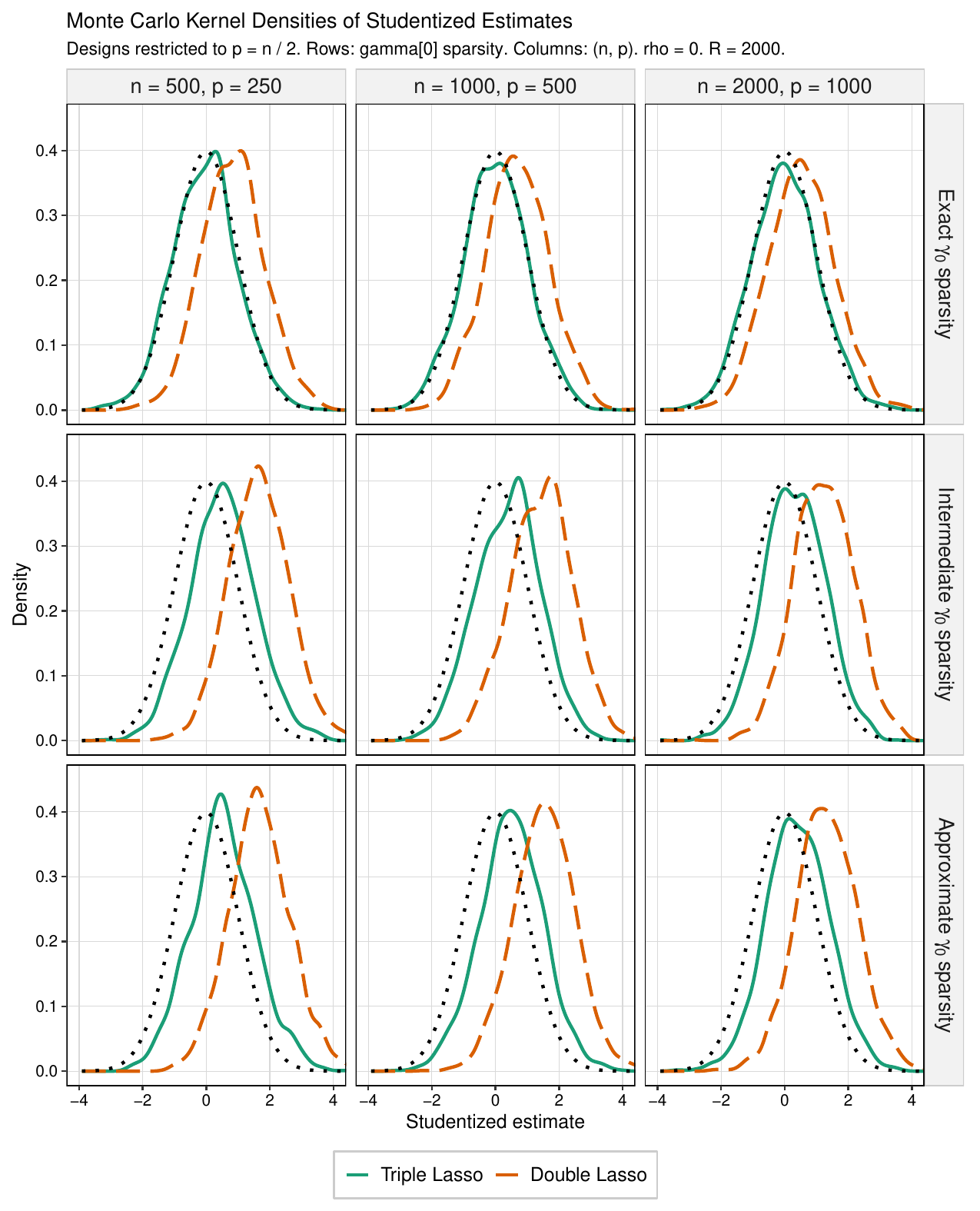}
\label{fig:sim-density-rho0}
\end{figure}

\newpage
\appendix

\part*{Appendix}
\section{Proofs for Results from Main Text}
\begin{proof}[Proof of Lemma \ref{lem: novel score}]
Simple calculations show that
$$
\frac{\partial\psi^{DL}(\beta_0,\bgamma_0,\bphi_0)}{\partial\bgamma} = \beta_0\E[(D-\bX^\top\bgamma_0)\bX] - \E[(Y - \bX^\top\bphi_0 - \beta_0(D - \bX^\top\bgamma_0))\bX] = \bzero_p,
$$
$$
\frac{\partial\psi^{DL}(\beta_0,\bgamma_0,\bphi_0)}{\partial\bphi} = -\E[(D-\bX^\top\bgamma_0)\bX] = \bzero_p,
\quad \frac{\partial^2\psi^{DL}(\beta_0,\bgamma_0,\bphi_0)}{\partial\bphi\partial\bphi^\top} = \bzero_{p\times p},
$$
$$
\frac{\partial^2\psi^{DL}(\beta_0,\bgamma_0,\bphi_0)}{\partial\bgamma\partial\bgamma^\top} = - 2\beta_0\E[\bX\bX^\top],
\quad \frac{\partial^2\psi^{DL}(\beta_0,\bgamma_0,\bphi_0)}{\partial\bgamma\partial\bphi^\top} = \E[\bX\bX^\top].
$$
Also,
\begin{align*}
\frac{\partial\psi^{ADJ}(\beta_0,\bgamma_0,\bphi_0,\bTheta_0)}{\partial\bgamma} 
& = \beta_0\E[\bX\bX^\top]\bTheta_0\E[\bX(D-\bX^\top\bgamma_0)] \\
&\quad - \E[\bX\bX^\top]\bTheta_0\E[\bX(Y - \bX^\top\bphi_0 - \beta_0(D - \bX^\top\bgamma_0))] = \bzero_p,
\end{align*}
$$
\frac{\partial\psi^{ADJ}(\beta_0,\bgamma_0,\bphi_0,\bTheta_0)}{\partial\bphi} = -\E[\bX\bX^\top]\bTheta_0 \E[\bX(D-\bX^\top\bgamma_0)] = \bzero_p,
$$
$$
\frac{\partial^2\psi^{ADJ}(\beta_0,\bgamma_0,\bphi_0,\bTheta_0)}{\partial\bgamma\partial\bgamma^\top} = -2\beta_0\E[\bX\bX^\top],
\quad \frac{\partial^2\psi^{ADJ}(\beta_0,\bgamma_0,\bphi_0,\bTheta_0)}{\partial\bphi\partial\bphi^\top} = \bzero_{p\times p},
$$
$$
\frac{\partial^2\psi^{ADJ}(\beta_0,\bgamma_0,\bphi_0,\bTheta_0)}{\partial\bgamma\partial\bphi^\top} = \E[\bX\bX^\top]\bTheta_0\E[\bX\bX^\top] = \E[\bX\bX^\top].
$$
In addition, by similar calculations, 
$$
\frac{\partial\psi^{ADJ}(\beta_0,\bgamma_0,\bphi_0,\bTheta_0)}{\partial\vec(\bTheta)} = \bzero_{p^2},
\quad \frac{\partial^2\psi^{ADJ}(\beta_0,\bgamma_0,\bphi_0,\bTheta_0)}{\partial\vec(\bTheta)\partial\vec(\bTheta)^\top} = \bzero_{p^2\times p^2},
$$
$$
\frac{\partial^2\psi^{ADJ}(\beta_0,\bgamma_0,\bphi_0,\bTheta_0)}{\partial\vec(\bTheta)\partial\bgamma^\top} = \bzero_{p^2\times p},
\quad \frac{\partial^2\psi^{ADJ}(\beta_0,\bgamma_0,\bphi_0,\bTheta_0)}{\partial\vec(\bTheta)\partial\bphi^\top} = \bzero_{p^2\times p}.
$$
Combining these identities gives the asserted claim.
\end{proof}

\begin{proof}[Proof of Theorem \ref{thm: asymptotic normality}]
Fix $k\in[K]$ and denote $\widehat\btheta_k := \widehat\bphi_k - \beta_0\widehat\bgamma_k$ and $\mathbb D_k := \{(\bX_i, D_i, Y_i)\}_{i\in I(-k)}$. Also, denote
\begin{align*}
& N_1 := \E_{I(k)}[\varepsilon_i\nu_i],\quad N_2 := - (\widehat\btheta_k - \btheta_0)^\top \E_{I(k)}[\bX_i\nu_i], \\
& N_3 := - (\widehat\bgamma_k - \bgamma_0)^\top \E_{I(k)}[\bX_i \varepsilon_i], \quad N_4 := (\widehat\btheta_k - \btheta_0)^\top \E_{I(k)}[\bX_i\bX_i^\top](\widehat\bgamma_k - \bgamma_0), \\
&N_5 := - \E_{I(k)}[\varepsilon_i \bX_i^\top]\widehat\bTheta_k\E_{I(k)}[\bX_i\nu_i],\quad N_6 :=  (\widehat\btheta_k - \btheta_0)^\top \E_{I(k)}[\bX_i\bX_i^\top]\widehat\bTheta_k \E_{I(k)}[\bX_i\nu_i], \\
& N_7 :=   (\widehat\bgamma_k - \bgamma_0)^\top \E_{I(k)}[\bX_i\bX_i^\top] \widehat\bTheta_k^\top \E_{I(k)}[\bX_i \varepsilon_i], \\
& N_8 := - (\widehat\btheta_k - \btheta_0)^\top \E_{I(k)}[\bX_i\bX_i^\top]\widehat\bTheta_k \E_{I(k)}[\bX_i\bX_i^\top](\widehat\bgamma_k - \bgamma_0).
\end{align*}
In addition, denote
\begin{align*}
& D_1 := \E_{I(k)}[\nu_i^2],\quad D_2:= D_3:= -(\widehat\bgamma_k - \bgamma_0)^\top \E_{I(k)}[\bX_i\nu_i], \\
& D_4 := (\widehat\bgamma_k - \bgamma_0)^\top \E_{I(k)}[\bX_i\bX_i^\top](\widehat\bgamma_k - \bgamma_0),\quad D_5 := -\E_{I(k)}[\nu_i\bX_i^\top]\widehat\bTheta_k\E_{I(k)}[\bX_i\nu_i], \\
& D_6 :=  (\widehat\bgamma_k - \bgamma_0)^\top\E_{I(k)}[\bX_i\bX_i^\top]\widehat\bTheta_k\E_{I(k)}[\bX_i\nu_i],
\  D_7 := (\widehat\bgamma_k - \bgamma_0)^\top\E_{I(k)}[\bX_i\bX_i^\top]\widehat\bTheta_k^\top\E_{I(k)}[\bX_i\nu_i], \\
& D_8 := -(\widehat\bgamma_k - \bgamma_0)^\top\E_{I(k)}[\bX_i\bX_i^\top]\widehat\bTheta_k\E_{I(k)}[\bX_i\bX_i^\top](\widehat\bgamma_k - \bgamma_0).
\end{align*}
Then
$$
\widehat\beta_{k}^{TL} - \beta_0 = \frac{N_1 + N_2 + N_3 + N_4  + N_5 + N_6 + N_7 + N_8}{D_1 + D_2 + D_3 + D_4  + D_5 + D_6 + D_7 + D_8}.
$$
We now proceed in two steps. In the first step, we show that
$$
N_1 + N_2 + N_3 + N_4 + N_5 + N_6 + N_7 + N_8 = \E_{I(k)}[\varepsilon_i\nu_i] + R_1 + R_2 + R_3
$$
for some $R_1$, $R_2$, and $R_3$ satisfying \eqref{eq: r1 and r2} and \eqref{eq: r3}. In the second step, we show that
$$
D_1 + D_2 + D_3 + D_4 + D_5 + D_6 + D_7 + D_8 = \E[\nu^2] + o_P(1).
$$
Combining these bounds and aggregating the estimators gives the asserted claims of the theorem.

\medskip
{\bf Step 1.}
Observe that
\begin{align*}
\E[N_2^2\mid \mathbb D_k] 
& \overset{(i)}{=} (\widehat\btheta_k - \btheta_0)^\top \E[\nu^2 \bX\bX^{\top}](\widehat\btheta_k - \btheta_0)/|I(k)| \\
& \overset{(ii)}{\lesssim} (\widehat\btheta_k - \btheta_0)^\top \E[\bX\bX^{\top}](\widehat\btheta_k - \btheta_0) / n \overset{(iii)}{\lesssim} s_{\bphi}\log p / n^2,
\end{align*}
where (i) follows from $\E[\nu\bX] = \bzero_p$, (ii) from Assumptions \ref{as: subsample} and \ref{as: bounded}, and (iii) from Lemma \ref{lem: l2 bounds}. Thus, $|N_2| \lesssim_P \sqrt{s_{\bphi}\log p}/n$ by Markov's inequality. Also, $|N_3| \lesssim_P \sqrt{s_{\bgamma}\log p}/n$ by the same argument. In addition,
\begin{align*}
&\E[N_5^2\mid \mathbb D_k, \{(\bX_i,D_i)\}_{i\in I(k)}]  \overset{(i)}{=} \E_{I(k)}[\nu_i\bX_i^\top]\widehat\bTheta_k^\top \E_{I(k)}[\bX_i\bX_i^\top]\widehat\bTheta_k \E_{I(k)}[\bX_i\nu_i] / |I(k)| 
\overset{(ii)}{\lesssim}_P s_{\bgamma}/n^2,
\end{align*}
where (i) follows from $\E[\varepsilon \mid \bX] = 0$ and Assumption \ref{as: bounded} and (ii) from Assumption \ref{as: subsample} Lemma \ref{lem: rudelson 3}. Thus, $|N_5| \lesssim_P \sqrt{s_{\bgamma}} / n$ by Markov's inequality. Moreover, denoting
$$
N_{6,1}:=(\widehat\btheta_k - \btheta_0)^\top(\E_{I(k)}[\bX_i\bX_i^\top] - \E[\bX\bX^\top])\widehat\bTheta_k\E_{I(k)}[\bX_i\nu_i],
$$
$$
N_{6,2}:= (\widehat\btheta_k - \btheta_0)^\top \E[\bX\bX^\top]\widehat\bTheta_k\E_{I(k)}[\bX_i\nu_i],
$$
we have $N_6 = N_{6,1} + N_{6,2}$. Here, $|N_{6,1}| \lesssim_P s_{\bgamma}\sqrt{s_{\bphi}}\log p / n^{3/2}$ by Lemma \ref{lem: rudelson}. Also, denoting
$$
N_{6,2,1}:=(\widehat\btheta_k - \btheta_0)^\top\E[\bX\bX^\top](\widehat\bTheta_k\E[\bX\bX^\top] - \bI_{\widehat T_k})\widehat\bTheta_k^\top\E[\bX\bX^\top](\widehat\btheta_k - \btheta_0) / n,
$$
$$
N_{6,2,2}:=(\widehat\btheta_k - \btheta_0)^\top\E[\bX\bX^\top]\bI_{\widehat T_k}\widehat\bTheta_k^\top\E[\bX\bX^\top](\widehat\btheta_k - \btheta_0) / n,
$$
we have
\begin{align*}
\E[N_{6,2}^2\mid \mathbb D_k]
& \overset{(i)}{=} (\widehat\btheta_k - \btheta_0)^\top\E[\bX\bX^\top]\widehat\bTheta_k\E[\nu^2\bX\bX^\top]\widehat\bTheta_k^\top\E[\bX\bX^\top](\widehat\btheta_k - \btheta_0) / |I(k)| \\
& \overset{(ii)}{\lesssim} (\widehat\btheta_k - \btheta_0)^\top\E[\bX\bX^\top]\widehat\bTheta_k\E[\bX\bX^\top]\widehat\bTheta_k^\top\E[\bX\bX^\top](\widehat\btheta_k - \btheta_0) / n 
\overset{(iii)}{=} N_{6,2,1} + N_{6,2,2},
\end{align*}
where (i) follows from $\E[\nu\bX]=\bzero_p$, (ii) from Assumptions \ref{as: subsample} and \ref{as: bounded}, and (iii) from the definitions of $N_{6,2,1}$ and $N_{6,2,2}$. Here,
\begin{align*}
|N_{6,2,1}|
& \overset{(i)}{\lesssim}_P \sqrt{s_{\bphi}\log p /n} \| (\widehat\bTheta_k\E[\bX\bX^\top] - \bI_{\widehat T_k})\widehat\bTheta_k^\top\E[\bX\bX^\top](\widehat\btheta_k - \btheta_0) \|_2 / n \\
& \overset{(ii)}{\lesssim}_P \sqrt{s_{\bphi}\log p /n} \| (\widehat\bTheta_k\E[\bX\bX^\top] - \bI_{\widehat T_k})\widehat\bTheta_k^\top\|_2 \| \sqrt{s_{\bphi}\log p /n} / n 
 \overset{(iii)}{\lesssim}_P s_{\bgamma}s_{\bphi}(\log p)^{3/2}/n^{5/2},
\end{align*}
where (i) follows from Lemma \ref{lem: l2 bounds} and Assumptions \ref{as: sparsity}.\ref{as: sparse eigenvalue} and \ref{as: sparsity}.\ref{as: sparse choice}, (ii) from the same argument as that in (i), and (iii) from Lemma \ref{lem: auxiliary inequalities} and Assumptions \ref{as: sparsity}.\ref{as: sparse choice} and \ref{as: nodewise lasso}.
Also,
\begin{align*}
|N_{6,2,2}|
&\overset{(i)}{\lesssim}_P \sqrt{s_{\bphi}\log p / n}\| \bI_{\widehat T_k}\widehat\bTheta_k^\top\E[\bX\bX^\top](\widehat\btheta_k - \btheta_0) \|_2 / n \\
&\overset{(ii)}{\lesssim}_P \sqrt{s_{\bphi}\log p / n} \| \bI_{\widehat T_k}\widehat\bTheta_k^\top\|_2 \sqrt{s_{\bphi}\log p / n} / n 
\overset{(iii)}{\lesssim}_P \sqrt{s_{\bgamma}}s_{\bphi}\log p / n^2,
\end{align*}
where (i) follows from Lemma \ref{lem: l2 bounds} and Assumptions \ref{as: sparsity}.\ref{as: sparse eigenvalue} and \ref{as: sparsity}.\ref{as: sparse choice}, (ii) from the same argument as that in (i), and (iii) from Assumptions \ref{as: sparsity}.\ref{as: sparse choice} and \ref{as: nodewise lasso}. Thus,
\begin{align*}
|N_6| & \overset{(i)}{\lesssim}_P s_{\bgamma}\sqrt{s_{\bphi}}\log p / n^{3/2} + \sqrt{s_{\bgamma}s_{\bphi}}(\log p)^{3/4}/n^{5/4} + s_{\bgamma}^{1/4}\sqrt{s_{\bphi}}\sqrt{\log p} / n  \overset{(ii)}{\lesssim}  s_{\bgamma}^{1/4}\sqrt{s_{\bphi}}\sqrt{\log p} / n,
\end{align*}
where (i) follows from the bounds above and Markov's inequality and (ii) from Assumption \ref{as: growth conditions}. Further,
\begin{align}
& \E[N_7^2\mid \mathbb D_k, \{(\bX_i,D_i\}_{i\in I(k)}] \nonumber\\
&\qquad \overset{(i)}{=} (\widehat\bgamma_k - \bgamma_0)^\top \E_{I(k)}[\bX_i\bX_i^\top] \widehat\bTheta_k^\top \E_{I(k)}[\bX_i\bX_i^\top]\widehat\bTheta_k \E_{I(k)}[\bX_i\bX_i^\top](\widehat\bgamma_k - \bgamma_0)  / |I(k)|\nonumber\\
&\qquad \overset{(ii)}{\lesssim} (\widehat\bgamma_k - \bgamma_0)^\top \E_{I(k)}[\bX_i\bX_i^\top] \widehat\bTheta_k^\top \E_{I(k)}[\bX_i\bX_i^\top]\widehat\bTheta_k \E_{I(k)}[\bX_i\bX_i^\top](\widehat\bgamma_k - \bgamma_0) / n \nonumber\\
&\qquad \overset{(iii)}{\lesssim}_P \| \widehat\bTheta_k \E_{I(k)}[\bX_i\bX_i^\top](\widehat\bgamma_k - \bgamma_0)\|_2^2 / n \label{eq: used in denominator}\\
&\qquad \overset{(iv)}{\lesssim} \| (\widehat\bTheta_k \E_{I(k)}[\bX_i\bX_i^\top] - \bI_{\widehat T_k})(\widehat\bgamma_k - \bgamma_0)\|_2^2 / n + \| \bI_{\widehat T_k}(\widehat\bgamma_k - \bgamma_0)\|_2^2 / n\nonumber\\
&\qquad \overset{(v)}{\lesssim}_P s_{\bgamma}\log p \|\widehat\bgamma_k - \bgamma_0\|_1^2 / n^2 + \|\widehat\bgamma_k - \bgamma_0\|_2^2 /  n
\overset{(vi)}{\lesssim}_P s_{\bgamma}^3 (\log p)^2 / n^3 + s_{\bgamma}\log p / n^2,\nonumber
\end{align}
where (i) follows from $\E[\varepsilon\mid \bX] = 0$ and Assumption \ref{as: bounded}, (ii) from Assumption \ref{as: subsample}, (iii) from the argument in the proof of Lemma \ref{lem: rudelson 3}, (iv) from the triangle inequality, (v) from Lemma \ref{lem: auxiliary inequalities}, H\"{o}lder's inequality, and Assumption \ref{as: sparsity}.\ref{as: sparse choice}, and (vi) from Assumption \ref{as: sparsity}.\ref{as: lasso estimation}. Thus, 
$
|N_7| \lesssim_P s_{\bgamma}^{3/2}\log p / n^{3/2} + \sqrt{s_{\bgamma}\log p} / n
$ 
by Markov's inequality. Finally, denoting
$$
N_{8,1} := (\widehat\btheta_k - \btheta_0)^\top \E_{I(k)}[\bX_i\bX_i^\top](\bI_{\widehat T_k} -\widehat\bTheta_k \E_{I(k)}[\bX_i\bX_i^\top])(\widehat\bgamma_k - \bgamma_0),
$$
$$
N_{8,2} := - (\widehat\btheta_k - \btheta_0)^\top \E_{I(k)}[\bX_i\bX_i^\top]\bI_{\widehat T_k}(\widehat\bgamma_k - \bgamma_0),
$$
we have $N_8 = N_{8,1} + N_{8,2}$. Here,
$|N_{8,1}| \lesssim_P s_{\bgamma}^{3/2}\sqrt{s_{\bphi}}(\log p)^{3/2} / n^{3/2}$
by Lemma \ref{lem: rudelson 2}. Also, 
\begin{align*}
|N_{8,2} + N_4| 
&\overset{(i)}{=} |(\widehat\btheta_k - \btheta_0)^\top\E_{I(k)}[\bX_i\bX_i^\top](\bgamma_0 - \bgamma_{0,\widehat T_k})| \\
&\overset{(ii)}{\leq} \|(\E_{I(k)}[\bX_i\bX_i^\top])^{1/2}(\widehat\btheta_0 - \btheta_0)\|_2  \|(\E_{I(k)}[\bX_i\bX_i^\top])^{1/2}(\bgamma_0 - \bgamma_{0,\widehat T_k})\|_2 \\
&\overset{(iii)}{\lesssim}_P \sqrt{s_{\bphi}\log p/n} \sqrt{(\bgamma_0 - \bgamma_{0,\widehat T_k})^\top\E[\bX\bX^\top](\bgamma_0 - \bgamma_{0,\widehat T_k})}, \\
&\overset{(iv)}{\lesssim}_P \sqrt{s_{\bphi}\log p /n}(\|\bar\bgamma_0 - \bar\bgamma_{0,\widehat T_k}\|_2
+ \sqrt{\E[|\bX^\top(\bgamma_0 - \bar\bgamma_0)|^2]}
+ \| \bgamma_0 - \bar\bgamma_0 \|_2),
\end{align*}
where (i) follows from $\bI_{\widehat T_k}\widehat\bgamma_k = \widehat \bgamma_k$, (ii) from the Cauchy-Schwarz inequality, (iii) from Lemma \ref{lem: l2 bounds} and Markov's and Jensen's inequalities, and (iv) from Lemma \ref{lem: approximation error}. Combining presented bounds and recalling the definition of $N_1$ completes this step.

\medskip
{\bf Step 2.}
Observe that $D_1 = \E[\nu^2] + o_P(1)$ by Chebyshev's inequality and Assumption \ref{as: moments}. Also, $D_2 = D_3 = O_P(\sqrt{s_{\bgamma}\log p} / n) = o_P(1)$, where the second equality follows from the same argument as that used to bound $N_2$ and $N_3$ and the third from Assumption \ref{as: growth conditions}. In addition, $D_4 = O_P(s_{\bgamma}\log p / n) = o_P(1)$, where the first equality follows from Lemma \ref{lem: l2 bounds} and the second from Assumption \ref{as: growth conditions}. Moreover,
\begin{align*}
|D_5|
&\overset{(i)}{\leq} \| \E_{I(k)}[\bX_{i,\widehat T_k}\nu_i \|_2\|\widehat \bTheta_k \E_{I(k)}[\bX_i\nu_i] \|_2 \overset{(ii)}{\lesssim}_P \sqrt{\widehat T_k/n}\sqrt{s_{\bgamma}/n}
\overset{(iii)}{\lesssim}_P s_{\bgamma}/n \overset{(iv)}{=}o(1),
\end{align*}
where (i) follows from the Cauchy-Schwarz inequality, (ii) from Markov's and Jensen's inequalities, Assumption \ref{as: bounded}, and Lemma \ref{lem: another l2 bound}, (iii) from Assumption \ref{as: sparsity}.\ref{as: sparse choice}, and (iv) from Assumption \ref{as: growth conditions}. Further,
$$
D_6 \overset{(i)}{\lesssim}_P s^{3/2}_{\bgamma}\log p / n^{3/2} +  s_{\bgamma}^{3/4}\sqrt{\log p}/n
\overset{(ii)}{=} o(1),
$$
where (i) follows from the same argument as that used to bound $N_6$ with $s_{\bphi}$ replaced by $s_{\bgamma}$ and (ii) from Assumption \ref{as: growth conditions}. Next,
\begin{align*}
|D_7|
&\overset{(i)}{\leq} \| \widehat\bTheta_k\E_{I(k)}[\bX_i\bX_i^\top](\widehat\bgamma_k - \bgamma_0) \|_2 \| \E_{I(k)}[\bX_{i,\widehat T_k}\nu_i]\|_2 \\
&\overset{(ii)}{\lesssim}_P \left(s^{3/2}_{\bgamma}\log p/n + \sqrt{s_{\bgamma}\log p / n}\right)\sqrt{s_{\bgamma}/n} \overset{(iii)}{=} o(1),
\end{align*}
where (i) follows from the Cauchy-Schwarz inequality, (ii) from Markov's and Jensen's inequalities, Assumption \ref{as: bounded}, and the same argument as that used to bound $N_7$ in \eqref{eq: used in denominator}, and (iii) from Assumption \ref{as: growth conditions}. Finally, $D_8 = D_{8,1} + D_{8,2}$, where 
$$
D_{8,1} := (\widehat\bgamma_k - \bgamma_0)^\top \E_{I(k)}[\bX_i\bX_i^\top](\bI_{\widehat T_k} -\widehat\bTheta_k \E_{I(k)}[\bX_i\bX_i^\top])(\widehat\bgamma_k - \bgamma_0),
$$
$$
D_{8,2} := - (\widehat\bgamma_k - \bgamma_0)^\top \E_{I(k)}[\bX_i\bX_i^\top]\bI_{\widehat T_k}(\widehat\bgamma_k - \bgamma_0),
$$
Then $D_{8,1} = O_P(s_{\bgamma}^{2}(\log p)^{3/2}/n^{3/2}) = o_p(1)$, where the first equality follows from the same argument as that used to bound $N_{8,1}$ and the second from Assumption \ref{as: growth conditions}. Also,
\begin{align*}
|D_{8,2}|
& \overset{(i)}{\leq} \| (\E_{I(k)}[\bX_i\bX_i^\top])^{1/2}(\widehat\bgamma_k - \bgamma_0) \|_2\| (\E_{I(k)}[\bX_i\bX_i^\top])^{1/2} \bI_{\widehat T_k} (\widehat\bgamma_k - \bgamma_0) \|_2 \\
& \overset{(ii)}{\lesssim}_P \sqrt{s_{\bgamma}\log p / n} \sqrt{s_{\bgamma}\log p / n} 
\overset{(iii)}{=} o(1), 
\end{align*}
where (i) follows from the Cauchy-Schwarz inequality, (ii) from Lemma \ref{lem: l2 bounds} and Assumptions \ref{as: sparsity}.\ref{as: sparse approximations}, \ref{as: sparsity}.\ref{as: lasso estimation}, \ref{as: sparsity}.\ref{as: sparse eigenvalue}, and \ref{as: sparsity}.\ref{as: sparse choice}, and (iii) from Assumption \ref{as: growth conditions}. Combining presented bounds completes this step.
\end{proof}

\begin{proof}[Proof of Lemma \ref{lem: higher-order neyman orthogonal score}]
We have $\tilde F(\beta_0,\tilde\bet_0) = 0$ since $F(\beta_0,\bet_0) = 0$ and $U(\beta_0,\bet_0) = \bzero_q$. Also, the first component of $\tilde U(\beta_0,\tilde\bet_0)$ is equal to $\bzero_q$ since $U(\beta_0,\bet_0) = \bzero_q$. The second component of $\tilde U(\beta_0,\tilde\bet_0)$ is equal to $\bzero_{q^2}$ by the definition of $\bbA_0$. The third component of $\tilde U(\beta_0,\tilde\bet_0)$ is equal to $\bzero_{q^k}$ by the definition of $B_0$ and the properties of the mode-wise tensor-matrix products. This yields the first asserted claim.

To prove the second asserted claim, note first that since $\tilde F(\beta,\tilde\bet)$ does not depend on $\bbA$, any component of the tensor $\nabla_{\tilde\bet}^m \tilde F(\beta_0,\tilde\bet_0)$ that includes derivatives with respect to components of $\bbA$ will be zero. Similarly, any component of this tensor that includes derivatives with respect to components of $\bbB$ will be zero as well since $U(\beta_0,\bet_0) = \bzero_q$. Thus, it suffices to prove that $\nabla_{\bet}^m \tilde F(\beta_0,\tilde\bet_0) = \bzero_q^{\otimes m}$ for all $m\in[k]$. To do so, fix any perturbation
$\bh=(\bh_{\bet}^\top,\bzero_{q^2}^\top,\bzero_{q^k}^\top)^\top
\in \R^{q+q^2+q^k}$, and expand $t\mapsto \tilde F(\beta_0,\tilde\bet_0 + t\bh)$ around $t = 0$.
To prepare for the expansion, we have
\begin{equation}\label{eq: tensor start}
U(\beta_0,\bet_0+t\bh_{\bet})=t \bbA_0^{\top}\bh_{\bet}+o(t)
\end{equation}
by the definition of $\bbA_0$, and so
\begin{align}
\bbB_0[(U(\beta_0,\bet_0+t\bh_{\bet}))^{\otimes k}]
&\overset{(i)}{=}
t^k\bbB_0[(\bbA_0^{\top}\bh_{\bet})^{\otimes k}]+o(t^k) \nonumber\\
& \overset{(ii)}{=} t^k (\bbA_0^{\otimes k}\bbB_0)[\bh_{\bet}^{\otimes k}] + o(t^k)
\overset{(iii)}{=} t^k \nabla_{\bet}^kF(\beta_0,\bet_0)[\bh_{\bet}^{\otimes k}] + o(t^k),\label{eq: tensor second step}
\end{align}
where (i) follows from \eqref{eq: tensor start}, (ii) follows from tensor properties, and (iii) from the definition of $\bbB_0$.
Now, since $F$ satisfies the $(k-1)$th order Neyman orthogonality condition, Taylor's expansion of $t\mapsto F(\beta_0,\bet_0 + t\bh_{\bet})$ around $t = 0$ yields
$$
F(\beta_0,\bet_0 + t\bh_{\bet}) = F(\beta_0,\bet_0)
 + \frac{t^k}{k!}\nabla_{\bet}^k F(\beta_0,\bet_0)[\bh_{\bet}^{\otimes k}] + o(t^k).
$$
Combining it with \eqref{eq: tensor second step} in turn gives
$$
\tilde F(\beta_0,\tilde\bet_0+t\bh_{\bet}) = F(\beta_0,\bet_0) + o(t^k).
$$
Hence, $\nabla_{\bet}^m \tilde F(\beta_0,\tilde\bet_0) = \bzero_q^{\otimes m}$ for all $m\in[k]$, yielding the second asserted claim and completing the proof of the lemma.
\end{proof}

\section{Auxiliary Lemmas}
\begin{lem}\label{lem: l2 bounds}
Under Assumption \ref{as: sparsity}, for all $k\in[K]$, we have
\begin{align*}
&\| (\E[\bX\bX^\top])^{1/2}(\widehat\bgamma_k - \bgamma_0) \|_2^2 \lesssim_P s_{\bgamma}\log p / n, \  \| (\E_{I(k)}[\bX_i\bX_i^\top])^{1/2}(\widehat\bgamma_k - \bgamma_0) \|_2^2 \lesssim_P s_{\bgamma}\log p / n, \\
&\| (\E[\bX\bX^\top])^{1/2}(\widehat\btheta_k - \btheta_0) \|_2^2 \lesssim_P s_{\bphi}\log p / n, \  \| (\E_{I(k)}[\bX_i\bX_i^\top])^{1/2}(\widehat\btheta_k - \btheta_0) \|_2^2 \lesssim_P s_{\bphi}\log p / n, 
\end{align*}
where we denoted $\widehat\btheta_k := \widehat\bphi_k - \beta_0\widehat\bgamma_k$.
\end{lem}

\begin{proof}[Proof of Lemma \ref{lem: l2 bounds}]
Fix $k\in[K]$ and observe that
$$
\| (\E[\bX\bX^\top])^{1/2}(\widehat\bgamma_k - \bgamma_0) \|_2 \leq \| (\E[\bX\bX^\top])^{1/2}(\widehat\bgamma_k - \bar\bgamma_0) \|_2 + \| (\E[\bX\bX^\top])^{1/2}(\bar\bgamma_0 - \bgamma_0) \|_2
$$
by the triangle inequality. Here,
\begin{align*}
\| (\E[\bX\bX^\top])^{1/2}(\widehat\bgamma_k - \bar\bgamma_0) \|_2^2 
& \overset{(i)}{=}(\widehat\bgamma_k - \bar\bgamma_0)^\top\E[\bX\bX^\top](\widehat\bgamma_k - \bar\bgamma_0) \\
& \overset{(ii)}{\lesssim}_P\|\widehat\bgamma_k - \bar\bgamma_0\|_2^2 \overset{(iii)}{\lesssim}_Ps_{\bgamma}\log p /n,
\end{align*}
where (i) follows from the definition of the $\ell_2$ norm, (ii) from Assumptions \ref{as: sparsity}.\ref{as: sparse vectors}, \ref{as: sparsity}.\ref{as: sparse estimators}, and \ref{as: sparsity}.\ref{as: sparse eigenvalue}, and (iii) from Assumption \ref{as: sparsity}.\ref{as: lasso estimation}. Also,
\begin{align*}
\| (\E[\bX\bX^\top])^{1/2}(\bar\bgamma_0 - \bgamma_0) \|_2^2 
& \overset{(i)}{=} (\bar\bgamma_0 - \bgamma_0)^\top\E[\bX\bX^\top](\bar\bgamma_0 - \bgamma_0) \\
& \overset{(ii)}{=}\E[| \bX^\top(\bar\bgamma_0 - \bgamma_0) |^2] \overset{(iii)}{\lesssim} s_{\bgamma}/n,
\end{align*}
where (i) follows from the definition of the $\ell_2$ norm, (ii) from noting that the transpose of a scalar is equal to the same scalar, and (iii) follows from Assumption \ref{as: sparsity}.\ref{as: sparse approximations}. Combining these bounds gives the first asserted claim. In turn, the second asserted claim follows from applying Markov's inequality conditional on $ \{(\bX_i, D_i, Y_i)\}_{i\in I(-k)}$ and using the first asserted claim.

Next, observe that
$$
\| (\E[\bX\bX^\top])^{1/2}(\widehat\bphi_k - \bphi_0) \|_2^2 \lesssim_P s_{\bphi}\log p / n, \  \| (\E_{I(k)}[\bX_i\bX_i^\top])^{1/2}(\widehat\bphi_k - \bphi_0) \|_2^2 \lesssim_P s_{\bphi}\log p / n
$$
by the same arguments as those used to prove the first and second asserted claims. The third and fourth asserted claims now follow via the triangle inequality since $\widehat\btheta_k - \btheta_0 = (\widehat\bphi_k - \bphi_0) - \beta_0(\widehat\bgamma_k - \bgamma_0)$ and $s_{\bphi} = s_{\bgamma} + s_{\btheta} \geq s_{\bgamma}$.
\end{proof}

\begin{lem}\label{lem: auxiliary inequalities}
Under Assumptions \ref{as: bounded}, \ref{as: subsample}, and \ref{as: nodewise lasso}, for all $k\in[K]$, we have
$$
\|\widehat\bTheta_k \E[\bX\bX^\top] - \bI_{\widehat T_k}\|_{\max}^2 \lesssim_P \log p/n
\quad\text{and}\quad \| \widehat\bTheta_k\E_{I(k)}[\bX_i\bX_i^\top] - \bI_{\widehat T_k} \|_\max \lesssim_P \log p/n.
$$
\end{lem}
\begin{proof}[Proof of Lemma \ref{lem: auxiliary inequalities}]
Fix $k\in[K]$ and observe that
\begin{equation}\label{eq: max inequality}
\| \E_{I(-k)}[\bX_i\bX_i^\top] - \E[\bX\bX^\top] \|_{\max}^2 \overset{(i)}{\lesssim}_P \log p/|I(-k)| \overset{(ii)}{\lesssim} \log p / n,
\end{equation}
where (i) follows from Assumption \ref{as: bounded}, Hoeffding's inequality, and the union bound and (ii) from Assumption \ref{as: subsample}. Thus,
\begin{align}
& \|\widehat\bTheta_k \E[\bX\bX^\top] - \bI_{\widehat T_k}\|_{\max}
 \overset{(i)}{\leq} \|\widetilde\bTheta_k \E[\bX\bX^\top] - \bI_{p}\|_{\max} \nonumber\\
& \qquad \overset{(ii)}{\leq} \|\widetilde\bTheta_k (\E[\bX\bX^\top] - \E_{I(-k)}[\bX_i\bX_i^\top])\|_{\max}
+ \| \widetilde\bTheta_k \E_{I(-k)}[\bX_i\bX_i^\top] - \bI_p \|_{\max} \nonumber\\
& \qquad \overset{(iii)}{\lesssim}_P \| \widetilde\bTheta_k \|_{\infty} \| \E[\bX\bX^\top] - \E_{I(-k)}[\bX_i\bX_i^\top] \|_{\max}
+ \sqrt{\log p/ n}  \overset{(iv)}{\lesssim}_P \sqrt{\log p/ n},\label{eq: inverse estimation bound 1}
\end{align}
where (i) follows from the definition of $\widehat\bTheta_k$, (ii) from the triangle inequality, (iii) from H\"{o}lder's inequality and Assumption \ref{as: nodewise lasso}, and (iv) from Assumption \ref{as: nodewise lasso} and the bound in \eqref{eq: max inequality}. This gives the first asserted claim.

To prove the second asserted claim, note that 
$$
\| \E_{I(k)}[\bX_i\bX_i^\top] - \E[\bX\bX^\top] \|_{\max}^2 \lesssim_P \log p / n
$$
by the same argument as that leading to \eqref{eq: max inequality} and proceed as in the proof of the first asserted claim.
\end{proof}

\begin{lem}\label{lem: another l2 bound}
Under Assumptions \ref{as: bounded}--\ref{as: nodewise lasso}, for all $k\in[K]$, we have 
$
\|\widehat\bTheta_k \E_{I(k)}[\bX_i\nu_i]\|_2^2 \lesssim_P s_{\bgamma}/n.
$
\end{lem}
\begin{proof}[Proof of Lemma \ref{lem: another l2 bound}]
Fix $k\in[K]$ and denote $\mathbb D_k := \{(\bX_i, D_i, Y_i)\}_{i\in I(-k)}$. Then
\begin{align*}
\E[ \|\widehat\bTheta_k \E_{I(k)}[\bX_i\nu_i]\|_2^2 \mid \mathbb D_k]
& \overset{(i)}{=} \E[ \E_{I(k)}[\nu_i\bX_i^\top]\widehat\bTheta_k^\top\widehat\bTheta_k\E_{I(k)}[\bX_i\nu_i] \mid \mathbb D_k] \\
& \overset{(ii)}{=} \tr(\E[\widehat\bTheta_k\E_{I(k)}[\bX_i\nu_i] \E_{I(k)}[\nu_i\bX_i^\top]\widehat\bTheta_k^\top \mid \mathbb D_k]) \\
& \overset{(iii)}{=} \tr(\widehat\bTheta_k\E[\nu^2 \bX\bX^\top]\widehat\bTheta_k^\top)/|I(k)| \\
& \overset{(iv)}{\lesssim} \tr(\widehat\bTheta_k\E[\bX\bX^\top]\widehat\bTheta_k^\top)/n 
\overset{(v)}{\lesssim}_P |\widehat T_k|/n
\overset{(vi)}{\lesssim}_P s_{\bgamma}/n,
\end{align*}
where (i) follows from the definition of the $\ell_2$ norm, (ii) from properties
of the trace operator, (iii) from $\E[\nu\bX] = \bzero_p$, (iv) from Assumptions
\ref{as: subsample} and \ref{as: bounded}, (v) from Assumptions \ref{as:
bounded} and \ref{as: nodewise lasso}, and (vi) from Assumption \ref{as:
sparsity}.\ref{as: sparse choice}. The asserted claim follows from combining
this bound with Markov's inequality.
\end{proof}

\begin{lem}\label{lem: rudelson}
Under Assumptions \ref{as: bounded}--\ref{as: growth conditions}, for all $k\in[K]$, we have
$$
|(\widehat\btheta_k - \btheta_0)^\top(\E_{I(k)}[\bX_i\bX_i^\top] - \E[\bX\bX^\top])\widehat\bTheta_k\E_{I(k)}[\bX_i\nu_i]| \lesssim_P s_{\bgamma}\sqrt{s_{\bphi}}\log p / n^{3/2},
$$
where we denoted $\widehat\btheta_k := \widehat\bphi_k - \beta_0\widehat\bgamma_k$.
\end{lem}
\begin{proof}[Proof of Lemma \ref{lem: rudelson}]
Fix $k\in[K]$ and denote $\mathbb D_k := \{(\bX_i, D_i, Y_i)\}_{i\in I(-k)}$. Then
\begin{align*}
& \E[\| \E_{I(k)}[\bX_{i,\widehat T_k}\bX_{i,\widehat T_k}^\top] - \E[\bX_{\widehat T_k}\bX_{\widehat T_k}^\top \mid \mathbb D_k] \|_2\mid \mathbb D_k] \\
& \qquad  \overset{(i)}{\lesssim} \frac{|\widehat T_k|\log p}{|I(k)|} +  \sqrt{\frac{|\widehat T_k|\| \E[\bX_{\widehat T_k}\bX_{\widehat T_k}^\top\mid \mathbb D_k] \|_2\log p}{|I(k)|}} 
 \overset{(ii)}{\lesssim_P} \sqrt{s_{\bgamma}\log p / n},
\end{align*}
where (i) follows from Lemma 6.2 in \cite{BCCH15} and Assumption \ref{as: bounded} and (ii) from Assumptions \ref{as: subsample}, \ref{as: sparsity}.\ref{as: sparse eigenvalue}, \ref{as: sparsity}.\ref{as: sparse choice}, and \ref{as: growth conditions}. Thus,
\begin{equation}\label{eq: rudelson main}
\| \E_{I(k)}[\bX_{i,\widehat T_k}\bX_{i,\widehat T_k}^\top] - \E[\bX_{\widehat T_k}\bX_{\widehat T_k}^\top \mid \mathbb D_k] \|_2 \lesssim_P \sqrt{s_{\bgamma}\log p / n}
\end{equation}
by Markov's inequality. Therefore,
$$
\|(\E_{I(k)}[\bX_i\bX_i^\top] - \E[\bX\bX^\top])\widehat\bTheta_k\E_{I(k)}[\bX_i\nu_i]\|_2 \lesssim_P \sqrt{s_{\bgamma}\log p / n} \sqrt{s_{\bgamma}/n}
$$
by Lemma \ref{lem: another l2 bound}. Also, $\|\widehat\btheta_k - \btheta_0\|_2 \lesssim_P \sqrt{s_{\bphi}\log p / n}$ by the triangle inequality and Assumptions \ref{as: sparsity}.\ref{as: sparse approximations} and \ref{as: sparsity}.\ref{as: lasso estimation}. Combining these bounds and using the Cauchy-Schwarz inequality yields the asserted claim.
\end{proof}

\begin{lem}\label{lem: rudelson 2}
Under Assumptions \ref{as: bounded}--\ref{as: growth conditions}, for all $k\in[K]$, we have
$$
|(\widehat\btheta_k - \btheta_0)^\top \E_{I(k)}[\bX_i\bX_i^\top](\bI_{\widehat T_k} -\widehat\bTheta_k \E_{I(k)}[\bX_i\bX_i^\top])(\widehat\bgamma_k - \bgamma_0)| \lesssim_P s_{\bgamma}^{3/2}\sqrt{s_{\bphi}}(\log p)^{3/2} / n^{3/2},
$$
where we denoted $\widehat\btheta_k := \widehat\bphi_k - \beta_0\widehat\bgamma_k$.
\end{lem}
\begin{proof}[Proof of Lemma \ref{lem: rudelson 2}]
Fix $k\in[K]$ and denote $\mathbb D_k := \{(\bX_i, D_i, Y_i)\}_{i\in I(-k)}$ and $\ba_k:=(\bI_{\widehat T_k} -\widehat\bTheta_k \E_{I(k)}[\bX_i\bX_i^\top])(\widehat\bgamma_k - \bgamma_0)$. Then
\begin{align*}
|\ba_k^\top\E_{I(k)}[\bX_i\bX_i^\top]\ba_k|
& \overset{(i)}{\leq}|\ba_k^\top(\E_{I(k)}[\bX_i\bX_i^\top] - \E[\bX\bX^\top])\ba_k|
+ |\ba_k^\top \E[\bX\bX^\top]\ba_k| \\
&\overset{(ii)}{\lesssim}_P \sqrt{s_{\bgamma}\log p / n}\|\ba_k\|_2^2 + \|\ba_k\|_2^2
\overset{(iii)}{\lesssim} \|\ba_k\|_2^2 \\
&\overset{(iv)}{\lesssim}|\widehat T_k|\| \bI_{\widehat T_k} - \widehat\bTheta_k\E_{I(k)}[\bX_i\bX_i^\top] \|_{\max}^2 \|\widehat\bgamma_k - \bgamma_0\|_1^2
\overset{(v)}{\lesssim}_P s_{\bgamma}^3(\log p)^2/n^2,
\end{align*}
where (i) follows from the triangle inequality, (ii) from \eqref{eq: rudelson main} in the proof of Lemma \ref{lem: rudelson} and Assumptions \ref{as: sparsity}.\ref{as: sparse eigenvalue} and \ref{as: sparsity}.\ref{as: sparse choice}, (iii) from Assumptions \ref{as: growth conditions}, (iv) from H\"{o}lder's inequality, and (v) from Lemma \ref{lem: auxiliary inequalities} and Assumptions \ref{as: sparsity}.\ref{as: sparse approximations}, \ref{as: sparsity}.\ref{as: lasso estimation}, and \ref{as: sparsity}.\ref{as: sparse choice}. The asserted claim follows from combining this bound with Lemma \ref{lem: l2 bounds} and noting that
$$
|(\widehat\btheta_k - \btheta_0)^\top\E_{I(k)}[\bX_i\bX_i^\top]\ba_k| \leq \|(\E_{I(k)}[\bX_i\bX_i^\top])^{1/2}(\widehat\btheta_k - \btheta_0)\|_2\|(\E_{I(k)}[\bX_i\bX_i^\top])^{1/2}\ba_k\|_2
$$
by the Cauchy-Schwarz inequality.
\end{proof}

\begin{lem}\label{lem: rudelson 3}
Under Assumptions \ref{as: bounded}--\ref{as: growth conditions}, for all $k\in[K]$, we have
$$
|\E_{I(k)}[\nu_i\bX_i^\top]\widehat\bTheta_k^\top \E_{I(k)}[\bX_i\bX_i^\top] \widehat\bTheta_k\E_{I(k)}[\bX_i\nu_i]| \lesssim_P s_{\bgamma}/n
$$
\end{lem}
\begin{proof}[Proof of Lemma \ref{lem: rudelson 3}]
Fix $k\in[K]$ and denote $\mathbb D_k := \{(\bX_i, D_i, Y_i)\}_{i\in I(-k)}$. Also, denote $\ba_k:=\widehat\bTheta_k\E_{I(k)}[\bX_i\nu_i]$. The asserted claim follows by noting that
\begin{align*}
|\ba_k^\top\E_{I(k)}[\bX_i\bX_i^\top]\ba_k|
& \overset{(i)}{\leq}|\ba_k^\top(\E_{I(k)}[\bX_i\bX_i^\top] - \E[\bX\bX^\top])\ba_k|
+ |\ba_k^\top \E[\bX\bX^\top]\ba_k| \\
&\overset{(ii)}{\lesssim}_P \sqrt{s_{\bgamma}\log p / n}\|\ba_k\|_2^2 + \|\ba_k\|_2^2
\overset{(iii)}{\lesssim} \|\ba_k\|_2^2 
\overset{(iv)}{\lesssim}_P s_{\bgamma}/n,
\end{align*}
where (i) follows from the triangle inequality, (ii) from \eqref{eq: rudelson main} in the proof of Lemma \ref{lem: rudelson} and Assumptions \ref{as: sparsity}.\ref{as: sparse eigenvalue} and \ref{as: sparsity}.\ref{as: sparse choice}, (iii) from Assumptions \ref{as: growth conditions}, and (iv) from Lemma \ref{lem: another l2 bound}.
\end{proof}

\begin{lem}\label{lem: approximation error}
Under Assumption \ref{as: sparsity}, for all $k\in[K]$, we have
$$
|(\bgamma_0 - \bgamma_{0,\widehat T_k})^\top\E[\bX\bX^\top](\bgamma_0 - \bgamma_{0,\widehat T_k})| \lesssim_P \|\bar\bgamma_0 - \bar\bgamma_{0,\widehat T_k}\|_2^2
+ \E[|\bX^\top(\bgamma_0 - \bar\bgamma_0)|^2]
+ \| \bgamma_0 - \bar\bgamma_0 \|_2^2.
$$
\end{lem}
\begin{proof}[Proof of Lemma \ref{lem: approximation error}]
Fix $k\in[K]$ and denote $\mathbb D_k := \{(\bX_i, D_i, Y_i)\}_{i\in I(-k)}$. Then
\begin{align*}
&|(\bgamma_0 - \bgamma_{0,\widehat T_k})^\top\E[\bX\bX^\top](\bgamma_0 - \bgamma_{0,\widehat T_k})|
\overset{(i)}{=}\E[| \bX^\top(\bgamma_0 - \bgamma_{0,\widehat T_k}) |^2\mid\mathbb D_k] \\
&\qquad \overset{(ii)}{\lesssim} \E[ |\bX^\top(\bar\bgamma_0 - \bar\bgamma_{0,\widehat T_k})|^2 \mid\mathbb D_k] + \E[| \bX^\top(\bgamma_0 - \bar\bgamma_0 - (\bgamma_{0,\widehat T_k} - \bar\bgamma_{0,\widehat T_k})) |^2 \mid\mathbb D_k] \\
&\qquad \overset{(iii)}{\lesssim} \E[| \bX^\top(\bar\bgamma_0 - \bar\bgamma_{0,\widehat T_k}) |^2\mid\mathbb D_k] + \E[| \bX^\top(\bgamma_0 - \bar\bgamma_0) |^2\mid\mathbb D_k] + \E[| \bX^\top(\bgamma_{0,\widehat T_k} - \bar\bgamma_{0,\widehat T_k}) |^2\mid\mathbb D_k] \\
&\qquad \overset{(iv)}{\lesssim}_P \|\bar\bgamma_0 - \bar\bgamma_{0,\widehat T_k}\|_2^2
+ \E[|\bX^\top(\bgamma_0 - \bar\bgamma_0)|^2]
+ \| \bgamma_0 - \bar\bgamma_0 \|_2^2,
\end{align*}
where (i) follows from noting that the transpose of a scalar is equal to the same scalar, (ii) from the triangle inequality, (iii) from the triangle inequality as well, and (iv) from Assumptions \ref{as: sparsity}.\ref{as: sparse vectors}, \ref{as: sparsity}.\ref{as: sparse eigenvalue}, and \ref{as: sparsity}.\ref{as: sparse choice}. The asserted claim follows.
\end{proof}

\section{Additional Simulation Results}\label{sec:additional-simulations}
Figures \ref{fig:sim-density-rho02}--\ref{fig:sim-density-rho08} contain kernel
density plots for the remaining values $\{0.2, 0.4, 0.6, 0.8\}$ of the parameter
$\rho$ dictating the correlation between controls. 

\begin{figure}[H]
\caption{Kernel densities of studentized estimates for $\rho=0.2$. The dotted line is the standard normal density.}
\centering
\includegraphics[width=\textwidth,height=0.88\textheight,keepaspectratio,trim=0 0 0 40,clip]{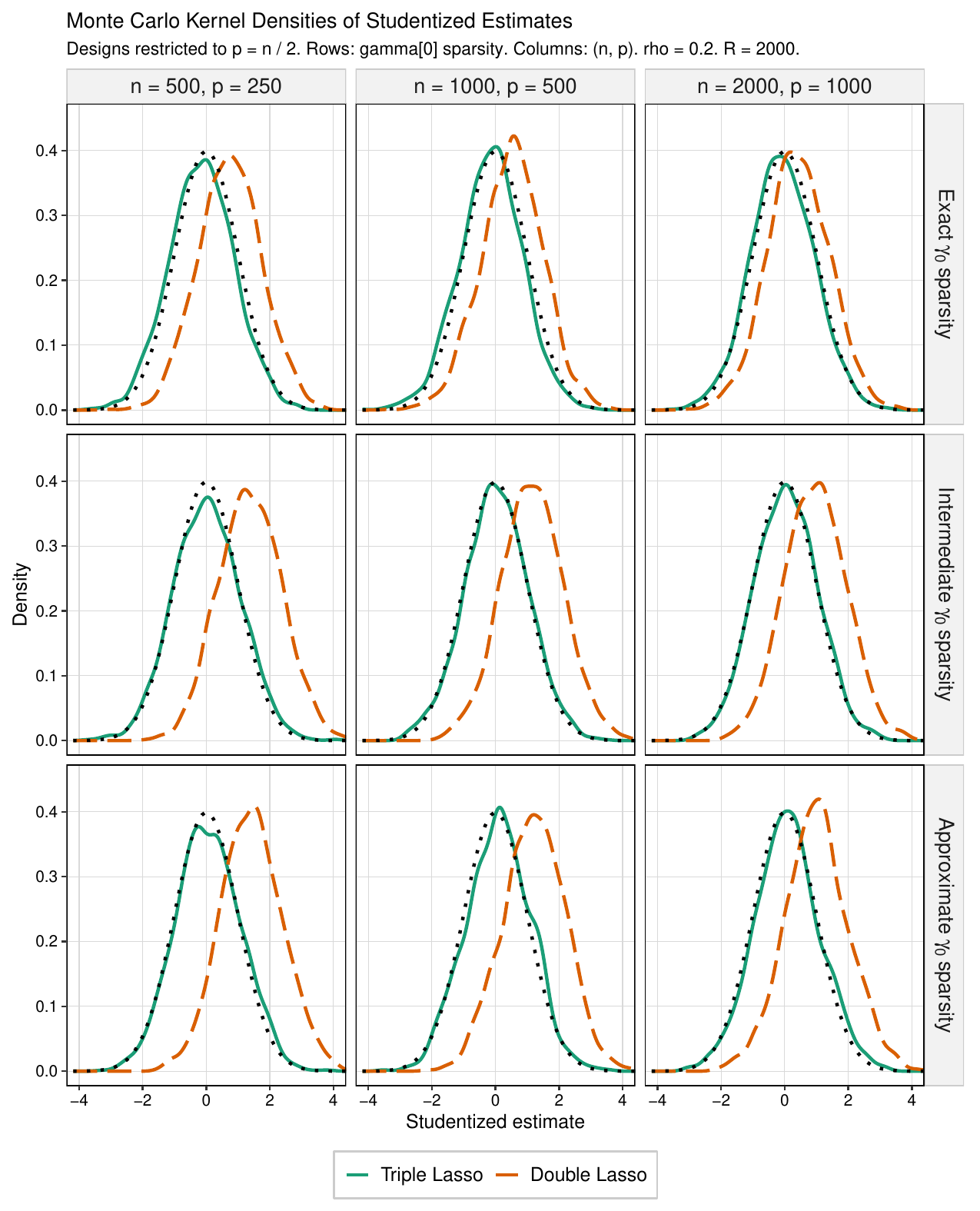}
\label{fig:sim-density-rho02}
\end{figure}

\begin{figure}[H]
\caption{Kernel densities of studentized estimates for $\rho=0.4$. The dotted line is the standard normal density.}
\centering
\includegraphics[width=\textwidth,height=0.88\textheight,keepaspectratio,trim=0 0 0 40,clip]{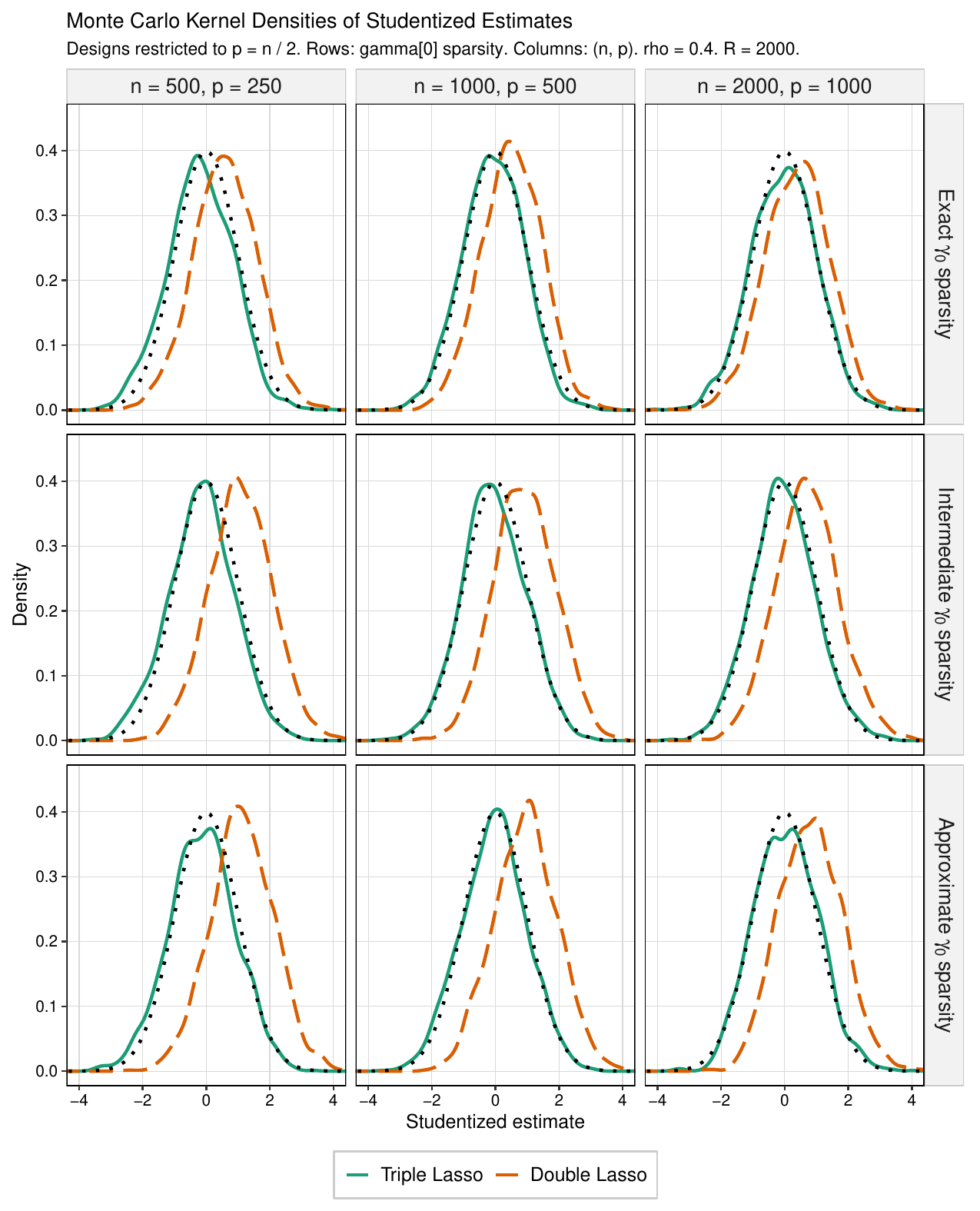}
\label{fig:sim-density-rho04}
\end{figure}

\begin{figure}[H]
\caption{Kernel densities of studentized estimates for $\rho=0.6$. The dotted line is the standard normal density.}
\centering
\includegraphics[width=\textwidth,height=0.88\textheight,keepaspectratio,trim=0 0 0 40,clip]{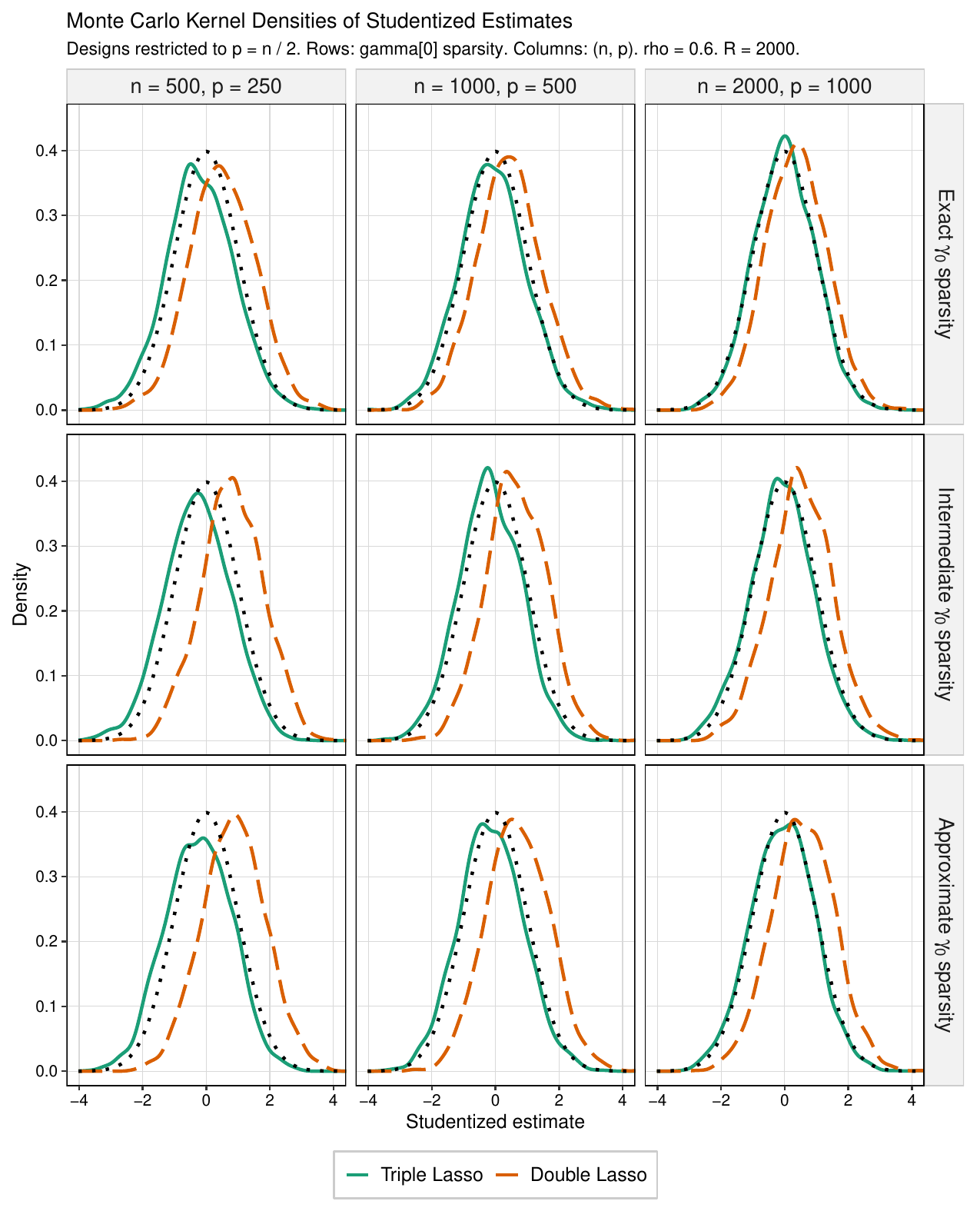}
\label{fig:sim-density-rho06}
\end{figure}

\begin{figure}[H]
\caption{Kernel densities of studentized estimates for $\rho=0.8$. The dotted line is the standard normal density.}
\centering
\includegraphics[width=\textwidth,height=0.88\textheight,keepaspectratio,trim=0 0 0 40,clip]{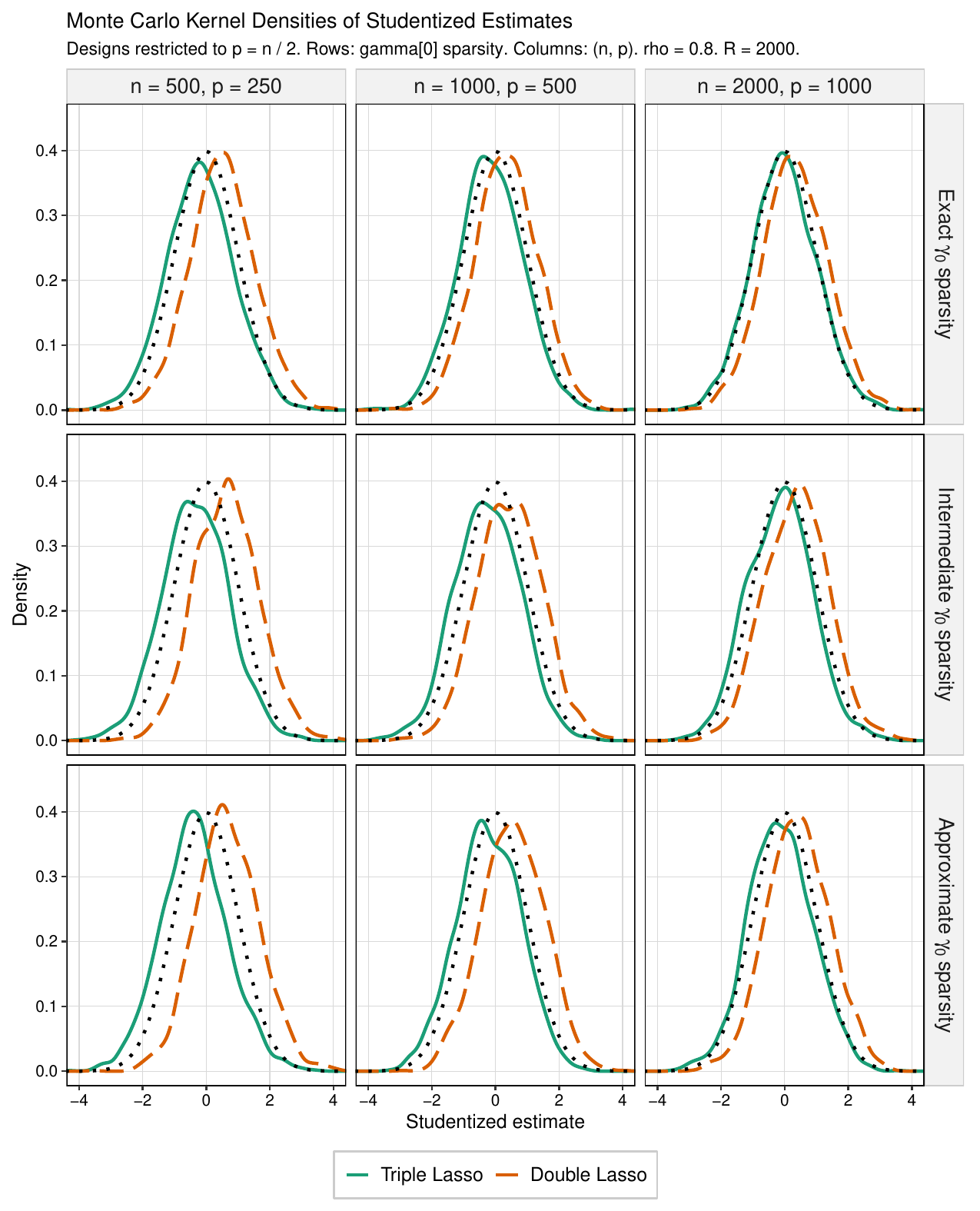}
\label{fig:sim-density-rho08}
\end{figure}

\end{document}